\documentclass[12pt]{article}

\usepackage{amsfonts}

\usepackage{multirow}
\usepackage{enumerate}
\usepackage{amsmath, dsfont, amscd, latexsym,  color}
\usepackage{amssymb}
\usepackage{setspace}
\usepackage{ifthen}

\newcommand{\handout}[5]{
  \noindent
  \begin{center}
    \framebox{
      \vbox{
        \hbox to 5.78in { {\bf } \hfill #2 }
        \vspace{4mm}
        \hbox to 5.78in { {\Large \hfill #5  \hfill} }
        \vspace{2mm}
        \hbox to 5.78in { {\em #3 \hfill #4} }
      }
    }
  \end{center}
  \vspace*{4mm}
}

\newtheorem{theorem}{Theorem}[section]
\newtheorem{corollary}[theorem]{Corollary}
\newtheorem{lemma}[theorem]{Lemma}

\newtheorem{prop}[theorem]{Proposition}
\newtheorem{definition}{Definition} \newtheorem{claim}[theorem]{Claim}
\newtheorem{fact}[theorem]{Fact}

 \newenvironment{proof}{\noindent
  \textbf{Proof:}\nopagebreak[2]}{$\qed$}

\newcommand{\qed}{\hfill\rule{7pt}{7pt} \medskip}

\newcommand{\Stab}{{\mathbb S}}

\newcommand{\eps}{\epsilon}

\newcommand{\Var}{\operatorname{Var}}

\newcommand{\sgn}{\mathrm{sign}}
\newcommand{\sign}{\mathrm{sign}}
\newcommand{\note}[1]{\marginpar{\tiny *note in TeX*}}

\newcommand{\ignore}[1]{}

\newcommand{\tref}[1]{Theorem~\ref{thm:#1}}

\newcommand{\cref}[1]{Corollary~\ref{cor:#1}}

\newcommand{\calN}{{\cal N}}
\newcommand{\calG}{{\cal G}}
\newcommand{\calC}{{\cal C}}

\renewcommand{\phi}{\varphi}

\newcommand{\AS}{\mathrm{AS}}
\newcommand{\NS}{\mathrm{NS}}
\newcommand{\GAS}{\mathrm{GAS}}
\newcommand{\GNS}{\mathrm{GNS}}
\newcommand{\GI}{\mathrm{GI}}

\newcommand{\bits}{\{-1,1\}}
\newcommand{\bn}{\bits^n}
\newcommand{\isafunc}{: \bn \rightarrow \bits}
\newcommand{\R}{\mathbb{R}}

\newcommand{\N}{{\cal N}}
\newcommand{\D}{{\cal D}}

\newcommand{\E}{\operatorname{{\bf E}}}
\newcommand{\Ex}{\mathop{{\bf E}\/}}
\renewcommand{\Pr}{\operatorname{{\bf Pr}}}

\newcommand{\littlesum}{\mathop{\textstyle \sum}}
\newcommand{\littlefrac}{\textstyle \frac}

\newcommand{\poly}{\mathrm{poly}}

\newcommand{\Inf}{\mathrm{Inf}}

\newcommand{\half}{{\textstyle \frac 1 2}}

\newcommand{\eqdef}{\stackrel{\textrm{def}}{=}}
\newcommand{\la}{\langle}
\newcommand{\ra}{\rangle}

\newcommand{\opt}{\mathsf{opt}}

\newcommand{\norm}[1]{\|#1\|}

\RequirePackage[colorlinks,pagebackref,linkcolor=blue,filecolor =blue,
citecolor = blue, urlcolor = blue]{hyperref}

\newcommand{\namedref}[2]{\hyperref[#2]{#1~\ref*{#2}}}
\newcommand{\claimref}[1]{\namedref{Claim}{#1}}
\newcommand{\corollaryref}[1]{\namedref{Corollary}{#1}}

\newcommand{\equationref}[1]{\namedref{Equation}{#1}}
\newcommand{\theoremref}[1]{\namedref{Theorem}{#1}}
\newcommand{\sectionref}[1]{\namedref{Section}{#1}}
\newcommand{\lemmaref}[1]{\namedref{Lemma}{#1}}

\renewcommand{\th}{\textsuperscript{th}\ }

\newcommand{\one}{\mathbf{1}}

\usepackage{nicefrac}
\newcommand{\nfrac}{\nicefrac}

\newcommand{\mper}{\,.}
\newcommand{\mcom}{\,,}

\newcommand{\commented}{yes}

\ifthenelse{\equal{\commented}{yes}}{
  \newcommand{\rnote}[1]{\footnote{{\bf [[Rocco: {#1}\bf ]]  }}}
  \newcommand{\inote}[1]{\footnote{{\bf [[Ilias: {#1}\bf ]]  }}}
  \newcommand{\lnote}[1]{\footnote{{\bf [[Li-Yang: {#1}\bf ]]  }}}
  
}

\ifthenelse{\equal{\commented}{no}}{
  \newcommand{\inote}[1]{}
  \newcommand{\lnote}[1]{}
  \newcommand{\rnote}[1]{}
}

\advance \topmargin by -\headheight
\advance \topmargin by -\headsep
\evensidemargin \oddsidemargin
\begin{document}

\title{Average sensitivity and noise sensitivity of polynomial
threshold functions}

\author{Ilias Diakonikolas\thanks{Department of Computer Science, Columbia University.
Email: {\tt ilias@cs.columbia.edu}. Research supported by NSF grant CCF-0728736, and by an Alexander S. Onassis Foundation
Fellowship. Part of this work was done while visiting IBM Almaden.} \and Prasad Raghavendra\thanks{Microsoft Research, New
England. Email: {\tt prasad@cs.washington.edu}. Part of the research done while at the University of Washington and visiting
Carnegie Mellon University supported by NSF CCF--0343672.} \and Rocco A. Servedio\thanks{Department of Computer Science,
Columbia University. Email: {\tt rocco@cs.columbia.edu}. Supported by NSF grants CCF-0347282, CCF-0523664 and CNS-0716245, and
by DARPA award HR0011-08-1-0069.} \and Li-Yang Tan\thanks{Department of Computer Science, Columbia University. Email: {\tt
liyang@cs.columbia.edu}. Supported by DARPA award no. HR0011-08-1-0069 and NSF Cybertrust grant no. CNS-0716245.}}

\setcounter{page}{0}

\maketitle

\thispagestyle{empty}

\begin{abstract}

  We give the first non-trivial upper bounds on the average
  sensitivity and noise sensitivity of degree-$d$ polynomial threshold
  functions (PTFs).  These bounds hold both for PTFs over the Boolean
  hypercube $\bn$ and for PTFs over $\R^n$ under the standard
  $n$-dimensional Gaussian distribution $\mathcal{N}(0,I_n)$. Our
  bound on the Boolean average sensitivity of PTFs represents progress
  towards the resolution of a conjecture of Gotsman and Linial
  \cite{GL:94}, which states that the symmetric function slicing the
  middle $d$ layers of the Boolean hypercube has the highest average
  sensitivity of all degree-$d$ PTFs.  Via the $L_1$ polynomial
  regression algorithm of Kalai et al. \cite{KKMS:08}, our bounds on
  Gaussian and Boolean noise sensitivity yield polynomial-time
  agnostic learning algorithms for the broad class of constant-degree
  PTFs under these input distributions.

  The main ingredients used to obtain our bounds on both average and
  noise sensitivity of PTFs in the Gaussian
  setting are tail bounds and anti-concentration bounds on
  low-degree polynomials in Gaussian random variables
  \cite{Janson:97,CW:01}.  To obtain our bound on the Boolean average
  sensitivity of PTFs, we generalize the ``critical-index'' machinery
  of \cite{Servedio:07cc} (which in that work applies to halfspaces,
  i.e. degree-1 PTFs) to general PTFs.  Together with the ``invariance
  principle'' of \cite{MOO:05}, this lets us extend our techniques
  from the Gaussian setting to the Boolean setting.  Our bound on
  Boolean noise sensitivity is achieved via a simple reduction from
  upper bounds on average sensitivity of Boolean PTFs to corresponding
bounds on noise sensitivity.

\end{abstract}

\ignore{


}

\newpage

\section{Introduction}

A degree-$d$ polynomial threshold function (PTF) over a domain $X
\subseteq \R^n$ is a Boolean-valued function $f: X \to \{-1,+1\},$
\[
f(x) = \sign(p(x_1,\dots,x_n))
\]
where $p : X \to \R$ is a degree-$d$ polynomial with real
coefficients.\ignore{The function $f$ is said to be a {\em
multilinear} PTF if the polynomial $p$ is multilinear.}  When $d=1$
polynomial threshold functions are simply linear threshold functions
(also known as halfspaces or LTFs), which play an important role in
complexity theory, learning theory, and other fields such as voting
theory. Low-degree PTFs (where $d$ is greater than 1 but is not too
large) are a natural generalization of LTFs which are also of
significant interest in these fields.

Over more than twenty years much research effort in the study of Boolean functions has been devoted to different notions of the
``sensitivity'' of a Boolean function to small perturbations of its input, see e.g. \cite{KKL:88, BshoutyTamon:96,
BourgainKalai:97, Friedgut:98, BKS:99, Shi:00, MosselOdonnell:03, MOO:05, OSSS:05, OdonnellServedio:07} and many other works.
In this work we focus on two natural and well-studied measures of this sensitivity, the ``average sensitivity'' and the ``noise
sensitivity.''  As our main results, we give the first non-trivial upper bounds on average sensitivity and noise sensitivity of
low-degree PTFs. These bounds have several applications in learning theory and complexity theory as we describe later in this
introduction.

We now define the notions of average and noise sensitivity in the
setting of Boolean functions $f: \{-1,1\}^n \to \{-1,1\}$.  (Our
paper also deals with average sensitivity and noise sensitivity of
functions $f: \R^n \to \{-1,1\}$ under the Gaussian distribution,
but the precise definitions are more involved than in the Boolean
case so we defer them until later.)

\subsection{Average Sensitivity and Noise Sensitivity}
\label{sec:asns}

The \emph{sensitivity} of a Boolean function $f: \{-1,1\}^n \to \{-1,1\}$ on an input $x \in \{-1,1\}^n$, denoted $s_f(x)$, is
the number of Hamming neighbors $y \in \{-1,1\}^n$ of $x$ (i.e. strings which differ from $x$ in precisely one coordinate) for
which $f(x) \neq f(y).$ The \emph{average sensitivity} of $f$, denoted $\AS(f)$, is simply $\E[s_f(x)]$ (where the expectation
is with respect to the uniform distribution over $\{-1,1\}^n$).  An alternate definition of average sensitivity can be given in
terms of the influence of individual coordinates on $f$.  For a Boolean function $f: \{-1,1\}^n \to \{-1,1\}$ and a coordinate
index $i \in [n]$, the \emph{influence of coordinate $i$ on $f$} is the probability that flipping the $i$-th bit of a uniform
random input $x \in \{-1,1\}^n$ causes the value of $f$ to change, i.e. $\Inf_i(f) = \Pr[f(x) \neq f(x^{\oplus i})]$ (where the
probability is with respect to the uniform distribution over $\{-1,1\}^n$). The sum of all $n$ coordinate influences,
$\littlesum_{i=1}^n \Inf_i(f)$, is called the \emph{total influence} of $f$; it is easily seen to equal $\AS(f).$ Bounds on
average sensitivity have been of use in the structural analysis of Boolean functions (see e.g.
\cite{KKL:88,Friedgut:98,Shi:00}) and in developing computationally efficient learning algorithms (see e.g.
\cite{BshoutyTamon:96,OdonnellServedio:07}).

The average sensitivity is a measure of how $f$ changes when a
single coordinate is perturbed.  In contrast, the noise sensitivity
of $f$ measures how $f$ changes when a random collection of
coordinates are all perturbed simultaneously. More precisely, given
a noise parameter  $0 \leq \eps \leq 1$ and a Boolean function $f:
\{-1,1\}^n \to \{-1,1\}$, the \emph{noise sensitivity of $f$ at
noise rate $\eps$} is defined to be
\[
\NS_\eps(f) = \Pr_{x,y}[f(x) \neq f(y)]
\]
where $x$ is uniform from $\{-1,1\}^n$ and $y$ is obtained from $x$
by flipping each bit independently with probability $\eps.$ Noise
sensitivity has been studied in a range of contexts including
Boolean function analysis, percolation theory, and computational
learning theory
\cite{BKS:99,KOS:04,MosselOdonnell:03,SchrammSteif:05,KOS:08}.

\subsection{Main Results:  Upper Bounds on Average Sensitivity and Noise Sensitivity}

\subsubsection{Boolean PTFs}

In 1994 Gotsman and Linial \cite{GL:94} conjectured that the symmetric function slicing the middle $d$ layers of the Boolean
hypercube has the highest average sensitivity among all degree-$d$ PTFs. Since this function has average sensitivity $\Theta(d
\sqrt{n})$ for every $1 \leq d \leq \sqrt{n}$, this conjecture implies (and is nearly equivalent to) the conjecture that every
degree-$d$ PTF $f$ over $\{-1,1\}^n$ has $\AS(f) \leq d \sqrt{n}.$

Our first main result is an upper bound on average sensitivity which
makes progress toward this conjecture:

\begin{theorem}
\label{thm:boolas} For any degree-$d$ PTF $f$ over $\{-1,1\}^n$, we
have
\[
\AS(f) \leq 2^{O(d)} \cdot\log n \cdot n^{1-1/(4d + 2)}.
\]

\end{theorem}

Using a completely different set of techniques, we also prove a
different bound which improves on \theoremref{thm:boolas} for $d
\leq 4$:

\begin{theorem} \label{thm:boolas2}
For any degree-$d$ PTF $f$ over $\{-1,1\}^n$, we have
\[
\AS(f) \leq 2 n^{1-1/2^d}.
\]

\end{theorem}

We give a simple reduction which translates any upper bound on
average sensitivity for degree-$d$ PTFs over Boolean variables into
a corresponding upper bound on noise sensitivity.  Combining this
reduction with Theorems~\ref{thm:boolas} and ~\ref{thm:boolas2}, we
establish:

\begin{theorem}
\label{thm:boolns} For any degree-$d$ PTF $f$ over $\{-1,1\}^n$ and
any $0 \leq \eps \leq 1$, we have \begin{eqnarray*} \NS_\eps(f)
&\leq& 2^{O(d)} \cdot \eps^{1/(4d + 2)} \log(1/\eps)\\
\NS_\eps(f) &\leq& O(\eps^{1/2^d}).
\end{eqnarray*}
\end{theorem}

\subsubsection{Gaussian PTFs}

Looking beyond the Boolean hypercube, there are well-studied notions of average sensitivity and noise sensitivity for
Boolean-valued functions over $\R^n$, where we view $\R^n$ as endowed with the standard multivariate Gaussian distribution
${\cal N}(0,I_n)$ \cite{Bog:98, MOO:05}. Let $\GAS(f)$ denote the Gaussian average sensitivity of a function $f: \R^n \to
\{-1,1\}$, and let $\GNS_\eps(f)$ denote the Gaussian noise sensitivity at noise rate $\eps.$  (See
\sectionref{sec:gaussianprelim} for precise definitions of these
quantities; here we just note that these are natural analogues of
their uniform-distribution Boolean hypercube counterparts defined
above.)  We prove an upper bound on Gaussian average sensitivity of
low-degree PTFs:

\begin{theorem}
\label{thm:gaussas} For any degree-$d$ PTF $f$ over $\R^n$, we have
\[ \GAS(f) \leq O(d^2\cdot \log n\cdot n^{1-1/2d}). \]
\end{theorem}

We remark that in the case of degree-$d$ \emph{multilinear} PTFs it is possible to obtain a slightly stronger bound of
$\GAS(f)\leq O(d\cdot \log n\cdot n^{1-1/2d})$ using our approach. We also prove an upper bound on the Gaussian noise
sensitivity of degree-$d$ PTFs:

\begin{theorem} \label{thm:gaussns} For any degree-$d$
PTF $f$ over $\R^n$ and any $0 \leq \eps \leq 1$, we have
\[ \GNS_\eps(f) \leq O(d\cdot \log^{1/2}(1/\eps)\cdot \eps^{1/2d}).\]
\end{theorem}

\ignore{ ROCCO -- seems to me we are getting the stronger bound of d
already in the non-multilinear case!

Here too, for degree-$d$ multilinear PTFs we have a stronger bound
of $\GNS_\eps(f)\leq O(d\cdot \log^{1/2} (1/\eps)\cdot
\eps^{1/2d})$.

}

\subsection{Application:  agnostically learning constant-degree PTFs
in polynomial time}

Our bounds on noise sensitivity, together with machinery developed
in \cite{KOS:04, KKMS:08, KOS:08}, let us obtain the first
efficient agnostic learning algorithms for low-degree polynomial
threshold functions. In this section we state our new learning
results; details are given in \sectionref{sec:learn}.

We begin by briefly reviewing the fixed-distribution agnostic learning framework that has been studied in several recent works,
see e.g. \cite{KKMS:08, KOS:08, BOW:08, GKK:08, KMV:08, SSS:09}. Let $\D_X$ be a (fixed, known) distribution over an example
space $X$ such as the uniform distribution over $\{-1,1\}^n$ or the standard multivariate Gaussian distribution ${\cal
N}(0,I_n)$ over $\R^n.$ Let $\calC$ denote a class of Boolean functions, such as the class of all degree-$d$ PTFs.  An
algorithm $A$ is said to be an \emph{agnostic learning algorithm for $\calC$ under distribution $\D_X$} if it has the following
property:  Let $\D$ be any distribution over $X \times \{-1,1\}$ such that the marginal of $\D$ over $X$ is $\D_X.$  Then if
$A$ is run on a sample of labeled examples drawn independently from $\D$, with high probability $A$ outputs a hypothesis $h: X
\to \{-1,1\}$ such that $\Pr_{(x,y) \sim \D}[h(x) \neq y] \leq \opt + \eps$, where $\opt = \min_{f \in \calC} \Pr_{(x,y) \sim
\D}[f(x) \neq y].$   In words, $A$'s hypothesis is nearly as accurate as the best hypothesis in $\calC.$

Kalai et al. \cite{KKMS:08} gave an $L_1$ polynomial regression
algorithm and showed that it can be used for agnostic learning. More
precisely, they showed that for a class $\calC$ of functions and a
distribution $\D$, if every function in $\calC$ has a low-degree
polynomial approximator (in the $L_2$ norm) under the marginal
distribution $\D_X$, then the $L_1$ polynomial regression algorithm
is an efficient agnostic learning algorithm for $\calC$ under
$\D_X.$  They used this $L_1$ polynomial regression algorithm
together with the existence of low-degree polynomial approximators
for halfspaces (under the uniform distribution on $\{-1,1\}^n$ and
the standard Gaussian distribution $\calN(0,I_n)$ on $\R^n$) to
obtain $n^{O(1/\eps^4)}$-time agnostic learning algorithms for
halfspaces under these distributions.

Using ingredients from \cite{KOS:04}, one can easily convert upper
bounds on Boolean noise sensitivity (such as
\theoremref{thm:boolns}) into results asserting the existence of
low-degree $L_2$-norm polynomial approximators under the uniform
distribution on $\{-1,1\}^n.$  We thus obtain the following agnostic
learning result (a more detailed proof is given in \sectionref{sec:learn}):

\begin{theorem}
\label{thm:learnuniform} The class of degree-$d$ PTFs is agnostically learnable under the uniform distribution on $\{-1,1\}^n$
in time \[ n^{2^{O(d^2)}(\log 1/\eps)^{4d+2}  / \eps^{8d+4}}.\] For $d \leq 4$, this bound can be improved to
$n^{O(1/\eps^{2^{d+1}})}.$
\end{theorem}

Similarly, using ingredients from \cite{KOS:08}, one can easily
convert upper bounds on Gaussian noise sensitivity (such as
\theoremref{thm:gaussns}) into results asserting the existence of
low-degree $L_2$-norm polynomial approximators under $\calN(0,I_n).$
This lets us obtain

\begin{theorem}
\label{thm:learngaussian} The class of degree-$d$ PTFs is agnostically learnable under any $n$-dimensional Gaussian
distribution in time $n^{(d/\eps)^{O(d)}}$.
\end{theorem}

For $\eps$ constant, these results are the first polynomial-time agnostic learning algorithms for constant-degree PTFs.

\subsection{Other applications}

The results and approaches of this paper have found other recent
applications beyond the agnostic learning results presented above;
we describe two of these below.

Gopalan and Servedio \cite{GS:09} have combined the average
sensitivity bound given by \theoremref{thm:boolas} with techniques
from \cite{LMN:93} to give the first sub-exponential time algorithms
for learning $AC^0$ circuits augmented with a small (but
super-constant) number of arbitrary threshold gates, i.e. gates that
compute arbitrary LTFs which may have weights of any magnitude.
(Previous work using different techniques \cite{JKS:02} could only
handle $AC^0$ circuits augmented with majority gates.)

In other recent work Diakonikolas et al. \cite{DSTW:09} have refined
the approach used to prove \theoremref{thm:boolas} to establish a
``regularity lemma'' for low-degree polynomial threshold functions.
Roughly speaking, this lemma says that any degree-$d$ PTF can be
decomposed into a constant number of subfunctions, almost all of
which are ``regular'' degree-$d$ PTFs.  \cite{DSTW:09} apply this
regularity lemma to extend the positive results on the existence of
low-weight approximators for LTFs, proved in \cite{Servedio:07cc},
to low-degree PTFs.

\paragraph{Related work.}  Simultaneously and independently of this work,
Harsha et al. \cite{HKM:09} have obtained very similar results on
average sensitivity, noise sensitivity, and agnostic learning of
low-degree PTFs using techniques very similar to ours.

\subsection{Techniques}
In this section we give a high-level overview of how Theorems~\ref{thm:boolas},~\ref{thm:gaussas} and ~\ref{thm:gaussns} are
proved.  (As mentioned earlier, \theoremref{thm:boolas2} is proved using completely different techniques; see
\sectionref{sec:booleanas2}.)  The arguments are simpler for the Gaussian setting so we begin with these.

\subsubsection{The Gaussian case}

We sketch the argument for the Gaussian noise sensitivity bound
\theoremref{thm:gaussns}; the Gaussian average sensitivity bound
\theoremref{thm:gaussas}, follows along similar lines

Let $f=\sign(p)$ where $p: \R^n \to \R$ is a degree-$d$ polynomial.  The Gaussian noise sensitivity $\GNS_\eps(f)$ of $f$ is
equal to $\Pr_{x,y}[f(x) \neq f(y)]$ where $x$ is distributed according to $\calN(0,I_n)$ and $y$ is an ``$\eps$-perturbed''
version of $x$ (see
\sectionref{sec:gaussianprelim} for the precise definition).  Intuitively,
the event $f(x) \neq f(y)$ can only take place if either
\begin{itemize}

\item $x$ lies close to the boundary of $p$, i.e. $|p(x)|$ is
``small'', or

\item $|p(x)-p(y)|$ is ``large''.

\end{itemize}

We use an anti-concentration result for polynomials in Gaussian
random variables, due to Carbery and Wright \cite{CW:01}, to show
that $|p(x)|$ is ``small'' only with low probability.  For the
second bullet, it turns out that $p(x)-p(y)$ can be expressed as a
low-degree polynomial in independent Gaussian random variables, and
thus we can apply tail bounds for this setting \cite{Janson:97} to
show that $|p(x)-p(y)|$ is ``large'' only with low probability. We
can thus argue that $\Pr_{x,y}[f(x) \neq f(y)]$ is low, and bound
the Gaussian noise sensitivity of $f.$  (We note that this
high-level explanation glosses over some significant technical
issues.  In particular, since we are dealing with general degree-$d$
PTFs which may not be multilinear, it is nontrivial to establish the
conditions that allow us to apply the tail bound; see the proof of
Claim~\ref{qnorm} in Section~\ref{sec:qnorm}.)

\subsubsection{The Boolean case}

One advantage of working over the Boolean domain $\{-1,1\}^n$ is
that without loss of generality we may consider only
\emph{multilinear} PTFs, where $f=\sign(p(x))$ for $p$ a multilinear
polynomial.  However, this advantage is offset by the fact that the
uniform distribution on $\{-1,1\}^n$ is less symmetric than the
Gaussian distribution; for example, every degree-1 PTF under the
Gaussian distribution $\N^n$ is equivalent simply to $\sign(x_1 -
\theta)$, but this is of course not true for degree-1 PTFs over
$\{-1,1\}^n$.  Our upper bound on Boolean average sensitivity uses
ideas from the Gaussian setting but also requires significant
additional ingredients.

An important notion in the Boolean case is that of a ``regular'' PTF; this is a PTF $f=\sign(p)$ where every variable \emph{in
the polynomial $p$} has low influence. (See \sectionref{sec:prel} for a definition of the influence of a variable on a
real-valued function; note that the definition from \sectionref{sec:asns} applies only for Boolean-valued functions.)  If $f$
is a regular PTF, then the ``invariance principle'' of \cite{MOO:05} tells us that $p(x)$ (where $x$ is uniform from
$\{-1,1\}^n$) behaves much like $p(\calG)$ (where $\calG$ is drawn from $\calN(0,I_n)$), and essentially the arguments from the
Gaussian case can be used.

It remains to handle the case where $f$ is not a regular PTF, i.e.
some variable has high influence in $p$.  To accomplish this, we
generalize the notion of the ``critical-index'' of a halfspace (see
\cite{Servedio:07cc,DGJ+09}) to apply to PTFs. We show that a
carefully chosen random restriction (one which fixes only the
variables up to the critical index -- very roughly speaking, only the
highest-influence variables -- and leaves the other ones free) has
non-negligible probability of causing $f$ to collapse down to a
regular PTF.  This lets us give a recursive bound on average
sensitivity which ends up being not much worse than the bound that can
be obtained for the regular case; see \sectionref{sec:booloverview}
for a detailed explanation of the recursive argument.

\subsection{Organization}
Formal definitions of average sensitivity and noise sensitivity
(especially in the Gaussian case), and tail bounds and
anticoncentration results for low degree
polynomials are presented in \sectionref{sec:prel}.
In \sectionref{sec:gaussianas}, we show an upper bound on the Gaussian
average sensitivity of PTFs (\theoremref{thm:gaussas}).  Upper bounds
on Gaussian noise sensitivity (\theoremref{thm:gaussns}) are obtained in the section that follows
(\sectionref{sec:gaussianns}).

The main result of the paper -- a bound on the Boolean average
sensitivity (\theoremref{thm:boolas}) -- is proved in
\sectionref{sec:booleanas}.  In \sectionref{sec:booleanas2}, an alternate bound for Boolean
average sensitivity that is better for degrees $d \leq 4$
(\theoremref{thm:boolas2}) is shown.  This is followed by a
reduction from Boolean average sensitivity bounds to corresponding
noise sensitivity bounds (\theoremref{thm:reduction}) in
\sectionref{sec:booleanns}.  We present the applications of these upper bounds to agnostic learning of PTFs in
\sectionref{sec:learn}.  Section~\ref{sec:discussion} concludes by proposing a
direction for future work towards the resolution of the Gotsman--Linial conjecture.

\section{Definitions and Background}
\label{sec:prel}
\label{sec:gaussianprelim}

\subsection{Basic Definitions} \label{sec:basicdef}

\ignore{
For a degree-$d$ polynomial $p : \bits^n \to \R$ we may write:
\[
p(x) = \sum_{|S| \leq d, S \subseteq [n]} \widehat{p}(S)x_S \quad
\text{~or simply~}\quad p(x) = \sum_{|S| \leq d} \widehat{p}(S)x_S.
\]

We will sometimes write $\|p\|$ to denote the $l_2$-norm of $p$,
i.e. $\|p\| \eqdef \sqrt{\sum_{S \subseteq [n]} \widehat{p}(S)^2}$.
}

In this subsection we record the basic notation and definitions used
throughout the paper. For $n \in \mathds{N}$, we denote by $[n]$ the
set $\{ 1, 2, \ldots, n\}$. We write $\N$ to denote the standard
univariate Gaussian distribution $\N(0,1).$

For a degree-$d$ polynomial $p : X \to \R$ we denote by $\|p\|_2$
its $l_2$ norm, $\|p\|_2 = \E_{x}[p(x)^2]^{1/2}$, where the intended
distribution over $x \in \R^n$ (which will always be either uniform
over $\{-1,1\}^n$, or the $\N^n$ distribution) will always be clear
from context.  We note that for multilinear $p$ the two notions are
always equal (see e.g. Proposition 3.5 of \cite{MOO:05}).

We now proceed to define the notion of influence for real-valued functions in a product probability space. Throughout this
paper we consider either the uniform distribution on the hypercube $\{\pm 1\}^n$ or the standard $n$-dimensional Gaussian
distribution in $\R^n$. However, for the sake of generality, we adopt this more general setting.

Let $(\Omega_1,\mu_1),\ldots,(\Omega_n,\mu_n)$ be probability spaces and let $(\Omega = \otimes_{i=1}^n \Omega_i, \mu =
\otimes_{i=1}^n \mu_i)$ denote the corresponding product space. Let $f : \Omega \to \R $ be any square integrable function on
$(\Omega, \mu)$, i.e. $f \in L^2(\Omega, \mu)$. The influence of the $i$th coordinate on $f$ \cite{MOO:05} is
\[ \Inf_i^{\mu}(f) \eqdef \E_\mu[\Var_{\mu_i}[f]]\] and the total influence
of $f$ is $\Inf^{\mu}(f)\eqdef\sum_{i=1}^n\Inf^{\mu}_i(f)$.

For a function $f : \bits^n \to \R$ over the Boolean hypercube endowed with the uniform distribution, the influence of variable
$i$ on $f$ can be expressed in terms of the Fourier coefficients of $f$ as,
$$\Inf_i(f) = \sum_{S\ni i} \widehat{f}(S)^2,$$
and as mentioned in the introduction it is easily seen that $\AS(f)
= \Inf(f)$ for Boolean-valued functions $f\isafunc$.

In this paper we are concerned with variable influences for
functions defined over $\bn$ under the uniform distribution, and
over $\R^n$ under $\N(0,I_n)$; we shall adopt the convention that
$\Inf_i(f)$ denotes the former and $\GI_i(f)$ the latter. We also
denote by $\GAS(f) = \sum_{i \in [n]} \GI_i(f)$ the Gaussian average
sensitivity.

Note that for a function $f : \R^n \to \{-1,1\}$, the Gaussian influence $\GI_i(f)$ can be equivalently written as: $ \GI_i(f)
= 2 \Pr_{x, x^{i}}[f(x) \neq f(x^{i})]$, where $x \sim \N^n$ and $x^{i}$ is obtained by replacing the $i$\th coordinate of $x$
by an independent random sample from $\N$.

We proceed to define the notion of noise sensitivity for
Boolean-valued functions in $(\R^n, \N^n)$. For the domain
$\{-1,1\}^n$, the notion has been defined already in the
introduction. (We remark that ``noise sensitivity'' can be defined
in a much more general setting and also for real-valued functions;
but such generalizations are not needed here.)

\begin{definition}[Gaussian Noise Sensitivity]  Given $f: \R^n \to \{-1,1\}$,
the ``Gaussian noise sensitivity of $f$ at noise rate $\eps \in
[0,1]$'' is
\[
\GNS_\eps(f) \eqdef        \Pr_{x,z}[f(x) \neq f(y)]; \]
where $x \sim \N^n$ and $y \eqdef (1-\eps)\,x + \sqrt{2\eps - \eps^2}\,z$ for an
independent Gaussian {\it noise} vector $z \sim \N^n$.
\end{definition}

\noindent {\bf Fourier and Hermite Analysis.} We assume familiarity
with the basics of Fourier analysis over the Boolean hypercube
$\{-1,1\}^n$.  We will also require similar basics of Hermite
analysis over the space $\R^n$ equipped with the standard
$n$-dimensional Gaussian distribution $\N^n$; a brief review is
provided in Appendix~\ref{ap:hermite}.

\subsection{Probabilistic Facts} \label{ssec:prob}

In this subsection, we record the basic probabilistic tools we use in our proofs.

We first recall the following well-known consequence of hypercontractivity (see e.g. Lecture 16 of~\cite{ODonnell:07} for the
boolean setting and~\cite{Bog:98} for the Gaussian setting):

\begin{theorem} \label{cor:deg-d-hyper}
Let $p: X \to \R$ be a degree-$d$ polynomial, where $X$ is either $\bn$ under the uniform distribution or $\R^n$ under $\N^n$,
and fix $q>2$. Then $$||p||^2_q \leq (q-1)^d ||p||^2_2.$$
\end{theorem}

We will need a concentration bound for low-degree polynomials over independent random signs or standard Gaussians. It can be
proved (in both cases) using Markov's inequality and hypercontractivity, see e.g. \cite{Janson:97,ODonnell:07,aushas09}.

\begin{theorem}[``degree-$d$ Chernoff bound''] \label{thm:deg-d-chernoff}
Let $p(x)$ be a degree-$d$ polynomial. Let $x$ be drawn either from the uniform distribution in $\bn$ or from $\N^n$. For any
$t > e^d$, we have \[ \Pr_x[|p(x)| \geq t \|p\|_2 ]\leq \exp(-\Omega(t^{2/d})). \]
\end{theorem}

\ignore{ \rnote{This is the version of the bound
    that Ryan stated in his email. With suitable modifications it even
    holds for non-multilinear polynomials.  He said it is in the
    Janson book, maybe also in the LeDoux paper.  He didn't mention
    the requirement that $t \geq e^d$, but I threw it in anyway -- I
    assume any time we use the theorem $t$ will be at least that large
    (otherwise the bound is trivial).} \inote{This bound holds as is
    even for the case that the constant term is not zero.  (The proof
    does not use any assumption about this constant.)  Recall we had
    the same issue in the Boolean setting.} }

The second fact is a powerful anti-concentration bound for
low-degree polynomials over Gaussian random variables.  (We note
that this result does not hold in the Boolean setting.)

\begin{theorem}[\cite{CW:01}]
\label{thm:carberywright} Let $p: \R^n \to \R$ be a degree-$d$ polynomial. Then for all $\eps>0$, we have
\[ \Pr_{x \sim \N^n}[|p(x)|\leq\eps \|p\|_2 ]\leq O(d\eps^{1/d}) .\]
\end{theorem}

We also make essential use of a (weak) anti-concentration property
of low-degree polynomials over the hypercube $\{-1,1\}^n$:

\begin{theorem} [\cite{DFKO06, aushas09}]
 \label{thm:aushas}
Let $p:\bn\to\R$ be a degree-$d$ polynomial with $\Var [p] \equiv \sum_{0<|S|\leq d} \widehat{p}(S)^2 =1$ and $\E[p] =
\widehat{p}(\emptyset) = 0.$ Then we have
$$\Pr[p(x)>1/2^{O(d)}]>1/2^{O(d)} \quad \quad \text{and hence}\quad \quad
\Pr[|p(x)| \geq 1/2^{O(d)}] > 1/2^{O(d)}.$$

\ignore{I put the following version on ice:
 Let
$p:\bn\to\R$ be a degree-$d$ multilinear polynomial with $\Var [p]
\equiv \sum_{0<|S|\leq d} \widehat{p}(S)^2 =1$. We have
\[ \Pr [p(x)>\nfrac{1}{2^{O(d)}}] > 1/2^{O(d)} .\]
}
\end{theorem}

\ignore{
\begin{lemma}
\label{lem:slab}
Let $p(x)=\sum_{|S| \leq d} \widehat{p}(S)x_S$ be a degree-$d$
multilinear polynomial with $\sum_{0 < |S| \leq d}
\widehat{p}(S)^2=1$.  Then,
\[ \Pr[|p(x)| \geq \nfrac{1}{2^{O(d)}} ] \geq 1/2^{O(d)}\]
\end{lemma}

\begin{proof}
  This follows directly from the above theorem. We apply the theorem
  twice: for the polynomial $q = p - \widehat{p}(\emptyset)$ and for
  $-q$. It is clear that for any value of the constant
  $\widehat{p}(\emptyset)$ we get the desired result.
\end{proof}

}

The following is a restatement of the invariance principle,
specifically Theorem
$3.19$ under hypothesis \textbf{H4} in \cite{MOO:05}.
\begin{theorem}[\cite{MOO:05}]
\label{thm:invariance} Let $p(x)=\sum_{|S| \leq d} \widehat{p}(S)x_S$ be a degree-$d$ multilinear polynomial with $\sum_{0 <
|S| \leq d} \widehat{p}(S)^2=1$. Suppose each variable $i\in[n]$ has low influence $\Inf_i(p) \leq \tau$, i.e. $\sum_{S\ni i}
\widehat{p}(S)^2 \leq \tau$.  Let $x$ be drawn uniformly from $\bn$ and $\mathcal{G} \sim \N^n$. Then,
\[ \sup_{t \in \R}|\Pr[p(x)\leq t] -
\Pr[p(\mathcal{G})\leq t]| \leq O(d\tau^{1/(4d+1)}).\]
 \end{theorem}

\section{Gaussian Average Sensitivity}
\label{sec:gaussianas}

\ignore{
It is well known that $\GI_i(f) = \sum_{J: J_i \neq 0}
\widetilde{f}(J)^2$ where $\widetilde{f}(J)$ is the $J$-th Hermite
coefficient of $f$, i.e. $f(x) = \sum_J \widetilde{f}(J)h_J(x)$;  we
will not need this, though.

\subsubsection{Tools}

We write $\|p\|^2_2$ to denote $\E_{\mathcal{G}}[p(\mathcal{G})^2].$
For multilinear $p(x)=\sum_\sigma p_\sigma x_\sigma$ this equals
$\sum_{\sigma}p_\sigma^2.$

\begin{theorem} \label{thm:deg-d-chernoff}[Concentration bound for
  low-degree multilinear polynomials over $\mathcal{G}$]: Let $p(x) =
  \sum_{\sigma} p_\sigma x_\sigma$ be a degree-$d$ multilinear
  polynomial with no constant term, i.e. $p_0=0$ (so
  $\|p\|=\Var[p(\mathcal{G})]$). Then for any $t>e^d$, we have
$$\Pr_{\mathcal{G}}[|p(\mathcal{G})|\geq t \|p\|_2]\leq 2
\exp(-\Omega(t^{2/d})).$$
\rnote{This is the version of the bound that Ryan stated in his
  email. With suitable modifications it even holds for non-multilinear
  polynomials.  He said it is in the Janson book, maybe also in the
  LeDoux paper.  He didn't mention the requirement that $t \geq e^d$,
  but I threw it in anyway -- I assume any time we use the theorem $t$
  will be at least that large (otherwise the bound is trivial).}
\end{theorem}
}

In this section we prove an upper bound on the Gaussian average
sensitivity of degree-$d$ PTFs (\theoremref{thm:gaussas}).

The following lemma, which relates the influence of a variable on
$f$ to its influence on the polynomial $p$, is central to the
argument.

\begin{lemma} \label{lemma:GIbound}
Let $p: \R^n \to \R$ be a degree-$d$ polynomial over Gaussian inputs
with $\Var[p]=1$ and let $f=\sign(p).$  Then for each $i \in [n]$,
$$ \GI_i(f) \leq O(d^2\cdot \GI_i(p)^{1/(2d)} \cdot
\log (1/\GI_i(p))).$$
\end{lemma}

\begin{proof} [of \lemmaref{lemma:GIbound}]
  Let $p(x)$ be a degree-$d$ polynomial where $\|p\|_2 = 1$.  For
  notational convenience let us fix $i=1$ and let $\tau = \GI_1(p)$.
  We may assume that $\tau < 1/4$ since otherwise the claimed bound
  holds trivially. We express $p(x)$ as a univariate polynomial in
  $x_1$ as follows,
\[ p(x) = p(x_1,\ldots,x_n) = \sum_{i=0}^d p_i(x_2,\ldots,x_n)\cdot
h_i(x_1) \] where $h_i(x_1)$ is the univariate degree-$i$ Hermite
polynomial.  Note that for any multi-index $S = (S_2,\ldots,S_n)\in
\mathbb{N}^{n-1}$ and $0\leq i\leq d$, we have $\widehat{p_i}(S) =
\widehat{p}(S')$ where $S' = (i,S_2,\ldots,S_{n})\in \mathbb{N}^n$.
As a result, using Parseval's identity for the Hermite basis, we
have that
\[ \norm{p}^2 = \sum_{i=0}^d \norm{p_i}^2. \]

We further have
\[ \norm{p_i}^2 = \sum_{S\in\mathbb{N}^{n-1}}\widehat{p_i}(S)^2 \quad
\text{ and }  \quad  \GI_i(p) = \sum_{S:S_i>0}\widehat{p}(S)^2. \]
Consequently the 2-norms of $p_1,\dots,p_d$ are ``small'' and the
2-norm of $p_0$ is ``large'':
\[ \sum_{i=1}^d \norm{p_i}^2 = \sum_{S : S_1 > 0}\widehat{p}(S)^2 =
\GI_1(p) = \tau \quad \text{ and }  \quad \norm{p_0}^2 = 1-\tau \geq 1/2.\]

Let $t=C^{d/2}\tau^{1/2}\log^{d/2}(1/\tau)$ and $\gamma =
d^2\cdot\tau^{1/2d}\log(1/\tau)$ where $C$ is an absolute constant
that will be defined later in \claimref{claim:pismall}.  We can
assume that $\gamma < 1/10$ since otherwise the bound of
Lemma~\ref{lemma:GIbound} holds trivially. For these values of $t$
and $\gamma$, the proof strategy is as follows:

\begin{itemize}

\item We use the ``small ball probability'' bound
  (\theoremref{thm:carberywright}) to argue that with high
  probability $p_0(g_2,\dots,g_n)$ is not too small: more precisely,
  $\Pr_{(g_2,\ldots,g_n) \sim \N^{n-1}}[|p_0(g_2,\dots,g_n)|\leq
  td(2ed \log (1/\gamma))^{d/2}] \leq O(\gamma)$ (see
  \claimref{claim:p0large}).

\item We use the concentration bound (\theoremref{thm:deg-d-chernoff})
  to argue that with high probability each $p_i(g_2,\dots,g_n)$,
  $i\in [d]$, is not too large: more precisely, $\Pr_{(g_2,\ldots,g_n)
    \sim \N^{n-1}}[|p_i(g_2,\dots,g_n)| \geq t] \leq O(\gamma)$ (see
  \claimref{claim:pismall}).

\item We use elementary properties of the $\N(0,1)$ distribution to
  argue that if $|a|\geq td(2ed \log (1/\gamma))^{d/2}$ and
  $|b_i|\leq t,$ then the function $\sign(a +
  \sum_{i=1}^db_ih_i(g_1))$ (a function of one $\N(0,1)$ random
  variable $g_1)$ is $O(\gamma)$-close to the constant function
  $\sign(a)$ (see \claimref{claim:univarbound}).

\item Thus we know that with probability at least $1 - O(\gamma)$ over
  the choice of $g_2,\dots,g_n$, we have
  $\Var_{g_1}[\sign(p(g_1,\dots,g_n))] \leq O(\gamma (1 - \gamma))
  \leq O(\gamma).$ For the remaining (at most) $O(\gamma)$ fraction of
  outcomes for $g_2,\dots,g_n$ we always have
  $\Var_{g_1}[\sign(p(g_1,\dots,g_n))] \leq 1$, so overall we get
  $\GI_1(\sign(p)) \leq O(\gamma).$

\end{itemize}

Thus, to complete the proof of the lemma, it suffices to prove the
three aforementioned claims.

\begin{claim}
\label{claim:p0large} With probability at least $1-O(\gamma)$ over
draws $(g_2,\ldots,g_n)\sim\N^{n-1}$, the polynomial
$p_0(g_2,\ldots,g_n)$ has magnitude at least $td(2ed \log
(1/\gamma))^{d/2}$.
\end{claim}

\begin{proof}
Applying \theoremref{thm:carberywright} to the polynomial
$p_0(x_2,\ldots,x_n)$ we get:
\[ \Pr_{g_2,\ldots,g_n}\left[ |p_0(g_2,\ldots,g_n)|\leq
  td(2ed \log
(1/\gamma))^{d/2}\right] \leq O(d)\cdot \left(\frac{td(2ed \log
(1/\gamma))^{d/2}}{\norm{p_0}}\right)^{1/d}.\] Recall that
$\norm{p_0}\geq \frac{1}{2}$, and so by our choice of $t$ and
$\gamma$ it follows that the right hand side is:
\[ O(d^{3/2})\cdot O(\tau^{1/2d}\log^{1/2}(1/\tau)\cdot \log^{1/2}(1/\gamma))
= O(\gamma).\]
\end{proof}

\begin{claim}
\label{claim:pismall}
  For each $i\in [d]$, the polynomial $p_i(g_2,\ldots,g_n)$ has
  magnitude larger than $t$ with probability at most
  $\gamma/d$. Therefore, the probability that any
  $p_i(g_2,\ldots,g_n)$ has magnitude larger than $t$ is at most
  $\gamma$.
\end{claim}

\begin{proof}
First note that since $\sum_{i=1}^d\norm{p_i}^2 = \tau$, certainly
for each $i\in [d]$ we have $\norm{p_i}\leq \sqrt{\tau}$. Therefore,
\[ |\E[p_i]| \leq \E[p_i^2]^{1/2} = \norm{p_i} \leq \sqrt{\tau}. \]
Let $p_i' = p_i - \E[p_i]$, so $\E[p'_i]=0.$  Applying
\theoremref{thm:deg-d-chernoff}, we get:
\[
\Pr_{g_2,\ldots,g_n}\left[|p_i'(g_2,\ldots,g_n)|>
  \frac{t-\sqrt{\tau}}{\norm{p_i'}}\cdot\norm{p_i'}\right] \leq 2\exp
\left(-\Omega\left(\left(\frac{t-\sqrt{\tau}}{\norm{p_i'}}\right)^{2/d}\right)\right).\]
Given our bound on $\norm{p_i'}\leq\norm{p_i}\leq \sqrt{\tau}$ and
choice of $t$, we know that the probability bound is at most
$2\exp(-\Omega(C\log(1/\tau)))$. For a sufficiently large absolute
constant $C$ this is at most $\exp(-4\log(1/\tau)) = \tau^4 \leq
\gamma/d$. To complete the proof note that if $|p'_i| \leq
t-\sqrt{\tau}$ then certainly $|p_i|\leq t$.
\end{proof}

We will need the following lemma in the proof of
Claim~\ref{claim:univarbound}:

\begin{lemma}
\label{lemma:hermitebound} The degree-$d$ Hermite polynomial
$h_d(x)$, $d\geq 1$, satisfies the following bound for all $x$:
\[ |h_d(x)| \leq (ed)^{d/2} \cdot \max\{1,|x|^d\}.\]
\end{lemma}

\begin{proof}  The lemma is immediate for $d=1.$ For $d\geq 2$, we
note that the polynomial $h_d(x)$ has at most $d$ terms, each of
which has coefficients of magnitude at most $\sqrt{d!}\leq
d^{d-1}/\sqrt{d!}$. This directly gives $|h_d(x)| \leq
(d^d/\sqrt{d!}) \cdot \max\{1,|x|^d\}.$  The claimed equality
follows easily from this using Stirling's approximation.
\end{proof}

\begin{claim}
\label{claim:univarbound} Suppose $|a|\geq td(2ed \log
(1/\gamma))^{d/2}$, $|b_i|\leq t$ for all $i\in [d]$, and $\gamma <
1/10$. Then,
\[ \Pr_{g_1\sim\N(0,1)}\left[\sgn(a+\sum_{i=1}^d b_i h_i(g_1)) \neq \sgn(a)\right]
\leq O(\gamma). \]
\end{claim}

\begin{proof}
If $\sgn(a+\sum_{i=1}^d b_i h_i(x)) \neq \sgn(a)$ then it has to be
the case that:
\[ \left|\sum_{i=1}^d b_ih_i(x)\right| \geq |a|. \]
By \lemmaref{lemma:hermitebound} we know that for all $x$, we have
\[ \left|\sum_{i=1}^d b_ih_i(x)\right|
\leq td\cdot \max_{1\leq i\leq d}|h_i(x)| < td(ed)^{d/2} \cdot
\max\{1,|x|^d\}.
\] Now if $|x|$ is at most $\sqrt{2\log(1/\gamma)}$,
since $\gamma < 1/10$ we have $\sqrt{2 \log (1/\gamma)}>1$ and hence
it follows that
\[  \left|\sum_{i=1}^d b_ih_i(x)\right|
\leq td(2ed\log(1/\gamma))^{d/2} \leq |a|. \] In other words, if
$\sgn(a+\sum_{i=1}^db_ih_i(x))$ differs from $\sgn(a)$, it must
necessarily be the case that $|x|\geq \sqrt{2\log(1/\gamma)}$. The
standard tail bound on Gaussians,
\[ \Pr_{g_1\sim\N(0,1)}[g_1<c]\leq
\frac{1}{\sqrt{2\pi}|c|}\exp(-c^2/2) \quad \quad \text{for $c < 0$},\]
completes the proof.
\end{proof}

The proof of \lemmaref{lemma:GIbound} is now complete.
\end{proof}

\ignore{
\begin{lemma} \label{lemma:one} Let $p: \R^n \to R$ be a
degree-$d$ multilinear polynomial over Gaussian inputs with
$\|p\|_2=1$ and let $f=\sign(p).$  Then for each $i \in [n]$,
$$ \GI_i(f) \leq O(d\cdot \GI_i(p)^{1/(2d)} \cdot
\log (1/\GI_i(p))).$$
\end{lemma}

\begin{proof}[of Lemma~\ref{lemma:one}]

For notational convenience let us fix $i=1.$  Let $\tau =
\GI_i(p)$. We may assume
that $\GI_i(p) = \tau < 1/4$ since
otherwise the claimed bound holds trivially.

Let us express $p(x)$ as a linear form in $x_1$, i.e.
\begin{eqnarray*}
p(x) = p(x_1,\dots,x_n) &=& a(x_2,\dots,x_n) + x_1 b(x_2,\dots,x_n).
\end{eqnarray*}

By assumption we have $\sum_S \widehat{b}(S)^2 = \tau$ and thus
$\sum_S \widehat{a}(S)^2 = 1-\tau \geq 1/2.$

The high-level idea of the proof is to show that with high
probability, a draw of $(g_2, \ldots, g_n) \sim \N^{n-1}$ causes
$\sign(p(g_1, g_2, \dots, g_n))$ (viewed as a function of $g_1$ only)
to be quite unbalanced.  Let $t = C^{d/2} \tau^{1/2} \log^{d/2}
(1/\tau)$ and $\gamma = d \cdot \tau^{1/(2d)} \log (1/\tau)$ where $C$
is an absolute constant that will be defined later in
\claimref{claim:second}. For these values of $t$ and $\gamma$, the
proof strategy is as follows:

\begin{itemize}

\item We use the ``small ball probability'' bound
  (Theorem~\ref{thm:carberywright}) to argue that with high
  probability, $a(g_2,\dots,g_n)$ is not too small: more precisely,
  $\Pr_{(g_2,\ldots,g_n) \sim \N^{n-1}}[|a(g_2,\dots,g_n)|\leq t
  \sqrt{2 \ln (1/\gamma)}] \leq O(\gamma)$ (see
  \claimref{claim:first}).

\item We use the concentration bound
  (Theorem~\ref{thm:deg-d-chernoff}) to argue that with high
  probability, $b(g_2,\dots,g_n)$ is not too large: more precisely,
  $\Pr_{(g_2,\ldots,g_n) \sim \N^{n-1}}[|b(g_2,\dots,g_n)| \geq t]
  \leq O(\gamma)$ (see \claimref{claim:second}).

\item We use elementary properties of the $N(0,1)$ distribution to
  argue that if $|a|\geq t \sqrt{2 \ln (1/\gamma)}$ and $|b|\leq t,$
  then the function $\sign(a + g_1 b)$ (a function of one $N(0,1)$
  random variable $g_1)$ is $O(\gamma)$-close to the constant function
  $\sign(a)$ (see \claimref{claim:third}).

\item Thus we know that with probability at least $1 - O(\gamma)$ over
  the choice of $g_2,\dots,g_n$, we have
  $\Var_{g_1}[\sign(p(g_1,\dots,g_n))] \leq O(\gamma (1 - \gamma))
  \leq O(\gamma).$ For the remaining (at most) $O(\gamma)$ fraction of
  outcomes for $g_2,\dots,g_n$ we always have
  $\Var_{g_1}[\sign(p(g_1,\dots,g_n))] \leq 1$, so overall we get
  $\GI_1(\sign(p)) \leq O(\gamma).$

\end{itemize}

Thus, to complete the proof of the lemma, it suffices to prove the
three aforementioned claims.

\begin{claim} \label{claim:first}
Let $a(g_2,\dots,g_n)$ be a degree-$d$ multilinear polynomial with
$\|a\|_2 \geq 1/2.$  Then
$\Pr_{g_2,\dots,g_n}[|a(g_2,\dots,g_n)|\leq t \sqrt{2 \ln
(1/\gamma)}] \leq O(\gamma)$. \end{claim}

\begin{proof} The claim follows directly from
  \theoremref{thm:carberywright} for the above choice of parameters.
  By applying \theoremref{thm:carberywright} to polynomial $a(x)$ we
  get
$$ \Pr_{g_2,\dots,g_n}\left[|a(g_2,\dots,g_n)| \leq \frac{t \sqrt{2 \ln
(1/\gamma)}}{\|a\|} \cdot \|a\|\right]  \leq O\left(d
\left(\frac{t \sqrt{2 \ln
(1/\gamma)}}{\|a\|} \right)^{1/d}\right) \mper $$
Recall that $\|a\| \geq \frac{1}{2}$, $t = C^{d/2}\tau^{1/2} \log^{d/2}
(1/\tau)$ and $\gamma = d \cdot \tau^{1/2d}
\log(1/\tau)$.  Substituting these in the previous
inequality yields the upper bound $O(\gamma)$ on the right hand side.
\end{proof}

\begin{claim} \label{claim:second} Let $b(g_2,\dots,g_n)$ be a
  degree-$(d-1)$ multilinear polynomial with $\|b\|_2^2 =\tau.$ There
  exists an absolute constant $C$ such that if $t = C^{d/2} \tau^{1/2}
  \log^{d/2}(1/\tau)$ then $\Pr_{g_2,\dots,g_n}[|b(g_2,\dots,g_n)|>t]
  \leq \gamma.$
\end{claim}

\begin{proof} Applying \theoremref{thm:deg-d-chernoff} on the
polynomial $b$ we get
$$ \Pr_{g_2,\dots,g_n}\left[|b(g_2,\dots,g_n)| > \frac{t}{\|b\|} \cdot
  \|b\| \right] \leq 2
\exp\left(-\Omega\left(\left(\frac{t}{\|b\|}\right)^{2/d}\right)\right)
\mper$$ Substituting the value of $\|b\| = \sqrt{\tau}$ and $t =
C^{d/2}\tau^{1/2} \log^{d/2} (1/\tau)$, the right hand side is upper
bounded by $2 \exp(-\Omega(C\log(1/\tau)))$.  For a sufficiently large
absolute constant $C$, this probability is less than $\exp(-4 \log
(1/\tau)) = \tau^4 \leq \gamma$.
\end{proof}

\begin{claim} \label{claim:third}  Fix any $t>0$ and $\gamma <{\frac 1 3}.$
If $|a|\geq t\sqrt{2\ln(1/\gamma)}$ and $|b| \leq t,$ then $\Pr_{g_1
\sim N(0,1)}[\sign(a+g_1b) \neq \sign(a)] < O(\gamma).$
\end{claim}

\begin{proof}
  We may assume wlog that $a,b>0.$ We have $\sign(a+g_1b) \neq
  \sign(a)$ only if $g_1 < -a/b$; since $-a/b \leq
  -\sqrt{2\ln(1/\gamma)}$, the standard bound

\[
\Pr_{g_1 \sim N(0,1)}[g_1 < c] \leq {\frac 1 {\sqrt{2 \pi}|c|}}
\exp(-c^2/2) \quad \quad \text{~for~}c<0
\]
(see e.g. p.6 of Durrett's \emph{Probability:  Theory and Examples})
gives the claim.
\end{proof}

The proof of Lemma~\ref{lemma:one} is now complete.
\end{proof}
}

We can now complete the proof of \theoremref{thm:gaussas}.

\begin{proof}[Proof of \theoremref{thm:gaussas}]
  Let us denote $\GI_i(p)$ by $\tau_i$ for $i \in [n]$. Note that
  since $p$ is of degree $d$, we have
\begin{equation}\label{eq:sumofinf}
\sum_{i \in [n]} \tau_i = \sum_{i\in [n]} \sum_{S\ni i} \widehat{p}(S)^2 =
\sum_{|S|\leq d}|S|\cdot\widehat{p}(S)^2 \leq d.
\end{equation}
Let $a_d(x) = d^2 x^{1/2d} \log(1/x)$.  By \lemmaref{lemma:GIbound} the
average sensitivity of $f$ can be bounded as
$$ \GAS(f) = \sum_{i \in [n]} \GI_i(f) \leq O(\sum_{i \in [n]}
a_d(\tau_i)). $$ The function $a_d(x)$ is monotone increasing and
concave in $[0, e^{-2d}]$.
In this light, we split the summation into terms greater than
$e^{-2d}$ and the rest. Let $S = \{i | \tau_i \geq e^{-2d}\}$ and $T =
[n]\setminus S$.  From \eqref{eq:sumofinf}, we have $|S| \leq
de^{2d}$.  Observe that for $n < (27d^2)^{2d}$, \theoremref{thm:gaussas}
holds trivially since $\GAS(f) \leq n \leq 27d^2n^{1-1/2d} \leq
27d^2n^{1-1/2d}\log n$.  Hence we may assume $n \geq (27d^2)^{2d}$, and
consequently $|T|$ is at least $n/2$.  Using concavity and
monotonicity of $a_d$, we can write
\begin{align*}
\sum_{i \in T} a_d(\tau_i)  \leq   |T| \cdot a_d\left((\sum_{i\in T}
\tau_i)/|T|\right)  \leq n a_d\left(\frac{2d}{n}\right) \leq O(d^2
n^{1-1/2d} \log n) \mper
\end{align*}
Therefore, the average sensitivity of $f$ is bounded by
\begin{align*}
\GAS(f) & = \sum_{i \in S} \GI_i(f) + \sum_{i \in T} \GI_i(f) \\
        & \leq   |S| + O(\sum_{i \in T} a_d(\tau_i)) \leq d e^{2d} +
O(d^2 n^{1-1/2d} \log n)   \mper
\end{align*}
For all $d \geq 1$ we have
\[
de^{2d} < e^{3d} < (3 d^{1/3})^{3d} \leq (27d^2)^d \leq n^{1/2},
\quad \quad \text{since~}n \geq (27d^2)^{2d}.\] Consequently we have
$\GAS(f) \leq n^{1/2} + O(d^2 n^{1-1/2d} \log n)  = O(d^2 n^{1-1/2d}
\log n)$, and the proof is complete.
\end{proof}

%
\section{Gaussian Noise Sensitivity}
\label{sec:gaussianns}

In this section we prove an upper bound on the noise sensitivity of
degree-$d$  PTFs.

\ignore{
In this section, we will prove an upper bound on the noise sensitivity
of a degree $d$ multilinear PTFs.  First, we setup some notation for
multilinear polynomials.
\begin{definition}
  A multilinear polynomial of degree $d$ is
  \[ p(x) = \sum_\sigma p_\sigma x_\sigma. \] Here each $\sigma$ is a
  sequence $(\sigma_1,\ldots,\sigma_n)$ in $\{0,1\}^n$ with at most
  $d$ ones, $x_\sigma$ denotes the monomial $\prod_{i=1}^n
  x_i^{\sigma_i}$, and $p_\sigma$ is a real coefficient.
\lnote{I guess we should keep this notation if we
  plan to generalize our Gaussian results to non-multilinear
  polynomials?}
\end{definition}
Specifically, we will be interested in $f = \sign(p)$ where $p$ is a
degree-$d$ multilinear polynomial given by $p(u) = \sum_{|\sigma|
  \leq d} p_\sigma x_\sigma$.  Without loss of generality, we may
assume that $\| p \|_2^2 = \sum_{|\sigma| \leq d} p^2_\sigma = 1$.

The proof strategy is as follows:

Our goal is to prove:

\begin{theorem} \label{thm:gaussianmultilinear-NS}
Let $p$ be an $n$-variable degree-$d$ multilinear polynomial over
Gaussian inputs. Then $\GNS_\eps(f)$, the Gaussian noise
sensitivity of $f=\sign(p)$ at noise rate $\eps$, is at most
$O(d \cdot \eps^{1/(2d)} \log(1/\eps))$.
\end{theorem}

\subsection{Background}

\begin{definition}
A multilinear polynomial of degree $d$ is
\[ p(x) = \sum_\sigma p_\sigma x_\sigma. \]
Here each $\sigma$ is a sequence $(\sigma_1,\ldots,\sigma_n)$ in
$\{0,1\}^n$ with at most $d$ ones, $x_\sigma$ denotes the monomial
$\prod_{i=1}^n x_i^{\sigma_i}$, and $p_\sigma$ is a real
coefficient.
\end{definition}

We define the Gaussian noise operator $T_\rho$ at noise rate $\rho$
in the usual way.  Given $f \in L^2(\R^n,\N^n)$, the noise operator
acts by
\[
(T_\rho f)(x) = \E_{y \sim \N^n}[f(\rho x + \sqrt{1 - \rho^2} y)]
\]
(This is just the Ornstein-Uhlenbeck operator with a slightly
different parameterization.)

\medskip

The \emph{noise stability of $f$ at noise rate $\rho$} is $
\Stab_\rho(f) = \la f, T_\rho f \ra$, where $\la f, g \ra$ denotes
$\E_{x \sim \N^n}[f(x) g (x)].$

\medskip

We now define the \emph{Gaussian noise sensitivity of $f$ at noise
  rate $\eps$}:

\noindent By definition of $T_{1-\eps}$, we have that
\begin{eqnarray}
  \GNS_\eps(f) &=& \half - \half \la f, T_{1-\eps} f \ra \nonumber\\
  &=& \half - \half \E_{x, z \sim \N^n}[f(x)f(y)], \qquad \text{where
    $y \eqdef (1-\eps)\,x + \sqrt{2\eps - \eps^2}\,z$}
  \nonumber\\
  & = & \Pr_{x,z}[f(x) \neq f(y)]; \label{eqn:NSprob}
\end{eqnarray}
i.e., $\GNS_\eps(f)$ is the probability that two
``$(1-\eps)$-correlated'' Gaussians land on opposite ``sides'' of
$f$.

\subsection{Tools}

We write $\|p\|^2_2$ to denote $\E_{\mathcal{G}}[p(\mathcal{G})^2].$
For multilinear $p(x)=\sum_\sigma p_\sigma x_\sigma$ this equals
$\sum_{\sigma}p_\sigma^2.$

\begin{theorem} \label{thm:deg-d-chernoff}[Concentration bound for
  low-degree multilinear polynomials over $\mathcal{G}$]: Let $p(x) =
  \sum_{\sigma} p_\sigma x_\sigma$ be a degree-$d$ multilinear
  polynomial with no constant term, i.e. $p_0=0$ (so $\|p\|_2^2 =
  \Var[p(\mathcal{G})]$). Then for any $t>e^d$, we have
$$\Pr_{\mathcal{G}}[|p(\mathcal{G})|\geq t \|p\|_2]\leq 2
\exp(-\Omega(t^{2/d})).$$
\rnote{This is the version of the bound that Ryan stated in his
  email. With suitable modifications it even holds for non-multilinear
  polynomials.  He said it is in the Janson book, maybe also in the
  LeDoux paper.  He didn't mention the requirement that $t \geq e^d$,
  but I threw it in anyway -- I assume any time we use the theorem $t$
  will be at least that large (otherwise the bound is trivial).}
\inote{This bound holds as is even for the case that the constant term
  is not zero.  (The proof does not use any assumption about this
  constant.)  Recall we had the same issue in the Boolean setting.}
\end{theorem}

\begin{theorem} \label{thm:cw} [Low-degree multilinear polys over
  $\mathcal{G}$ have low small ball probabilities]: There is a
  universal constant $C>0$ such that for any degree $d$ multilinear
  polynomial $p(x)=\sum_{\sigma} p_\sigma x_\sigma$ and any $\eps>0$,
  we have
\[ \Pr[|p(\mathcal{G})|\leq\eps]\leq Cd(\eps/\|p(\mathcal{G})\|_2)^{1/d}. \]
\end{theorem}
}

\begin{proof}[of \theoremref{thm:gaussns}]
Let $f = \sign(p)$, where $p=p(x_1,\dots,x_n)$ is a degree-$d$
polynomial with $\E_{x \sim \N^n}[p(x)^2]^{1/2} = \| p \|_2 = 1$.
Recall that $\GNS_{\eps}(f)$ equals  $\Pr_{x,z}[f(x) \neq f(y)]$
where $x \sim \N^n$, $z \sim \N^n$; $x$ and $z$ are independent; and
$y = \alpha x + \beta z$, with $\alpha \eqdef 1-\eps$ and $\beta =
\sqrt{2\eps-\eps^2}$.

We can assume wlog that $\eps \leq 2^{-2d-1}$, since otherwise the theorem trivially holds.

\ignore{  OLD DECOMPOSITION FOR MULTILINEAR CASE:

We can write
\[p(x) = \sum_{|S| \leq d} \widehat{p}(S) x_S\] and
\[p(y) = p(\alpha x + \beta z) = \sum_{|S| \leq d} \widehat{p}(S)
(\alpha x + \beta z)_S .\]

END OLD DECOMPOSITION FOR MULTILINEAR CASE}

Let us define the function
\[q(x,z) = p(x) - p(y) .\]
Note that $q$ is a degree-$d$ polynomial over $2n$ variables.

Fix a real number $t^{\ast} > 0$.  It is easy to see that $f(x)
\neq f(y)$ only if at least one of the following two events hold:
\[ (\text{Event } \mathcal{E}_1) \quad |p(x)| \leq t^{\ast} \qquad \text{ OR
} \qquad (\text{Event } \mathcal{E}_2) \quad |q(x,z)| \geq t^{\ast}.
\] We will upper bound the probability of these two events for a
carefully chosen $t^{\ast}$.
 We will bound the probability of the event
$\mathcal{E}_1$ using Carbery-Wright (\theoremref{thm:carberywright}), the probability of event $\mathcal{E}_2$ using the tail
bound for degree-$d$ polynomials (\theoremref{thm:deg-d-chernoff}) and then apply a union bound.

The choice of $t^{\ast}$ will be dictated by
\theoremref{thm:deg-d-chernoff}. More precisely, to apply
\theoremref{thm:deg-d-chernoff}, a bound on $\|q\|_2$ is needed.  To
this end, we show the following claim:

\begin{claim} \label{qnorm}
We have $\| q\|_2  = O(d \cdot \sqrt{\eps})$.
\end{claim}

The proof of this claim is somewhat involved and is deferred to
Section~\ref{sec:qnorm}.

\ignore{ OLD PROOF OF Q-NORM BOUND FOR MULTILINEAR CASE:

\begin{proof}
  Note that $q$ is multilinear, hence its $l_2$ norm is just the
  square root of the sum of its squared coefficients.  Let us expand
  the term $(\alpha x + \beta z)_S$ that appears in $p(y)$ with
  coefficient $\widehat{p}(S)$.  For notational convenience we assume
  that $S = [k]$, for
  some $k \leq d$.  We have that
\[ (\alpha x + \beta z)_S = \prod_{i=1}^k (\alpha x_i + \beta
z_i) = \sum_{T \subseteq [k]} \alpha^{|T|} x_{T} \beta^{k-|T|}
z_{\bar{T}} = \alpha^k x_{[k]} + \sum_{T \subsetneq [k]} \alpha^{|T|}
x_{T} \beta^{k-|T|} z_{\bar{T}} \]

Now, we observe that in the polynomial $q$ the monomial $x_S =
x_{[k]}$ will have coefficient $\widehat{p}(S)\cdot (1 - \alpha^k) = O(d \cdot
\eps) \cdot \widehat{p}(S)$. The monomials $x_{T} z_{\bar{T}}$ are all
distinct, even across different $S$'s (this is important since
it means there are no cancelations across different $S$'s).
For a given $S = [k]$,
we have that the sum of the squares of the coefficients of all these
monomials is exactly
\[\widehat{p}(S)^2 \cdot \sum_{i=0}^{k-1} \binom{k}{j}\alpha^{2j}
\beta^{2(k-j)} = \widehat{p}(S)^2 \cdot [(\alpha^2 + \beta^2)^k -
\alpha^{2k}] = \widehat{p}(S)^2 \cdot [1 - (1-\eps)^{2k}] = \widehat{p}(S)^2
\cdot O(d \cdot \eps).\]

By summing up across all $S$'s and taking square root we get the claim.
\end{proof}

END OF OLD PROOF OF Q-NORM BOUND FOR MULTILINEAR CASE}

Fix $t^{\ast} = \Theta (d \sqrt{\eps} \log^{d/2} (1/\eps))$.  By \theoremref{thm:carberywright}, we have:
\[ \Pr_{x \sim \N^n} [|p(x)| \leq t^{\ast}]  = O(d \cdot
(t^{\ast})^{1/d}) = O(d \cdot \eps^{1/(2d)} \cdot
\log^{1/2}(1/\eps)) .
\]

Since both $x$ and $y$ are individually distributed according to $\N^n$, we have $\E[q(x,z)] = \E[p(x)-p(y)] = 0.$  By
\theoremref{thm:deg-d-chernoff} and \claimref{qnorm}, we get
\[ \Pr_{x,z \sim \N^{2n}} \left[|q(x,z)| \geq \frac{t^{\ast}}{\|q\|_2}
\cdot \|q\|_2\right] \leq 2
\exp\left(-\Omega\left(\left(\frac{t^{\ast}}{\|q\|_2}\right)^{2/d}\right)\right)
\leq \eps \mper \] Hence, by a union bound the noise sensitivity is
$O(d \cdot \eps^{1/(2d)} \cdot \log^{1/2}(1/\eps))$.  This completes
the proof of \theoremref{thm:gaussns}.
\end{proof}

\subsection{Proof of Claim~\ref{qnorm}}~\label{sec:qnorm}
Let $p:\R^n \to \R$ be a degree-$d$ polynomial over independent standard Gaussian random variables. Let us assume that $\| p
\|_2 =1$ and that $\eps \leq 2^{-2d-1}$. We will show that $$\|q\|_2 = O(d \sqrt{\eps}).$$

It will be convenient for the proof to express $p$ in an appropriate orthonormal basis. Let $p(x) = \sum_{S \in \mathcal{S}}
\widehat{p}(S) H_S(x)$ be its Hermite expansion; $\mathcal{S}$ is a family of multi-indices where each $\{H_S\}_{S \in
\mathcal{S}}$ has degree at most $d$. By orthonormality of the basis we have that $$\|p\|_2^2 = \sum_{S\in \mathcal{S}}
\widehat{p}(S)^2.$$ Note that $q(x,z) = \sum_{S\in \mathcal{S}} \widehat{p}(S) \big( H_S(x)-H_S(y) \big)$ and
$$q^{2}(x,z) = \sum_{S\in \mathcal{S}} \widehat{p}^2(S) \big( H_S(x)-H_S(y) \big)^2 + \sum_{S, T\in \mathcal{S}, S \neq T} \widehat{p}(S) \widehat{p}(T) \big ( H_S(x)-H_S(y))(H_T(x)-H_T(y)) \big).$$
Let us denote the second summand in the above expression by $q'(x,z)$. We will first show that
$$\E_{x,z} [q'(x,z)] = 0.$$
By linearity of expectation we can write
\[\E_{x,z}[q'(x,z)] =  \sum_{S, T\in \mathcal{S}, S \neq T} \widehat{p}(S) \widehat{p}(T) \E_{x,z}\Big[ \big( H_S(x)-H_S(y))(H_T(x)-H_T(y)) \big) \Big] = 0.\]
Hence, it suffices to show that for all $S \neq T$ we have
\[ \E_{x,z}\Big[ \big( H_S(x)-H_S(y))(H_T(x)-H_T(y)) \big) \Big] = 0.\]
By orthogonality of the Hermite basis, and the fact that $y$ is distributed according to $\N^n$, the above expression equals
\[ -\E_{x,z}\Big[ H_S(x) H_T(y)\Big] - \E_{x,z}\Big[ H_S(y) H_T(x)\Big].\]

Thus, the desired result follows from the following lemma:

\begin{lemma} \label{lem:orthog}
For all $S \neq T$ it holds
\[ \E_{x,z} \big[H_S(x) H_T(y)\big] = 0.\]
\end{lemma}

\begin{proof}
Since $S \neq T$, it suffices to prove the result for univariate Hermite polynomials. The result for the multivariate case then
follows by independence. That is, for $x_1, z_1 \in \N(0,1)$ and $s \neq t \in [d]$, we need to show that
\[ \E_{x_1, z_1} \big[ h_s(x_1) h_t(\alpha x_1 + \beta z_1) \big] =0 .\]
Since $\alpha^2+\beta^2=1$, we have that the joint distribution of $(x_1, \alpha x_1 +\beta z_1)$ is identical to the joint
distribution of $(\alpha x_1 + \beta z_1, x_1)$, and thus we can assume wlog that $s>t$. Since $h_t(\alpha x_1 + \beta z_1)$ is
a degree-$t$ polynomial in $x_1, z_1$ it can be written in the form
\[ \sum_{i,j=0}^{t} c_{ij} h_i(x_1)h_j(z_1)\]
for some real coefficients $c_{ij}$. Hence, by linearity of expectation and independence, the desired expectation is
\[ \sum_{i,j=0}^{t} c_{ij} \E [h_i(x_1) h_s(x_1)] \cdot \E [h_j(z_1)]\]
which equals $0$ by orthogonality of the Hermite basis.
\end{proof}

At this point, we need the following claim whose proof is deferred to the following subsection:

\begin{claim} \label{clm:mv-hermite}
Let $H_d(x)$ be a degree-$d$ multivariate Hermite polynomial. Then
$$\|H_d(x) - H_d(y) \|_2  = O(d \cdot \sqrt{\eps}).$$
\end{claim}
Repeated applications of Claim~\ref{clm:mv-hermite} now yield
\begin{eqnarray*}
\E_{x,z}[q^2] &=& \sum_{S\in \mathcal{S}} \widehat{p}^2(S) \E_{x,z}\big[\big( H_S(x)-H_S(y) \big)^2\big]\\
              &\leq& \sum_{S\in \mathcal{S}} \widehat{p}^2(S) \cdot O(d^2 \cdot \eps) = O(d^2 \cdot \eps)
\end{eqnarray*}
concluding the proof.

\subsubsection{Proof of Claim~\ref{clm:mv-hermite}}

We can assume wlog that
\[ H_d(x) = \prod_{i=1}^j h_{k_i}(x_i) \]
where $j \in [d]$, $k_i \geq 1$, and $\sum_{i=1}^j k_i = d.$

For $i \in [j]$, we denote by $\Delta h_{k_i}(x_i, y_i) = h_{k_i}(y_i) - h_{k_i}(x_i)$. Then we can write
\begin{eqnarray*}
H_d(y) &=& \prod_{i=1}^j h_{k_i}(y_i) = \prod_{i \in [j]} \big( h_{k_i}(x_i) + \Delta h_{k_i}(x_i, y_i) \big) \\
       &=&  H_d(x) + \sum_{\emptyset \neq I \subseteq [j]} \prod_{i \in I}\Delta h_{k_i}(x_i, y_i) \cdot  \prod_{i \in [j]\setminus I} h_{k_i}(x_i) .\\
\end{eqnarray*}
We will need the following claim whose proof lies in the next subsection:
\begin{claim} \label{cl:uv-hermite}
Let $h_d(x)$ be a degree-$d$ univariate Hermite polynomial. Then
$$ \| \Delta h_d(x,y) \|_2 =    \|h_d(x)-h_d(y) \|_2  \leq  8 \sqrt{d} \cdot \sqrt{\eps}.$$
\end{claim}
The triangle inequality for norms combined with independence now yields
\[ \|H_d(x) - H_d(y) \|_2 \leq   \sum_{\emptyset \neq I \subseteq [j]} \prod_{i \in I} \| \Delta h_{k_i}(x_i, y_i) \|_2 \cdot  \prod_{i \in [j]\setminus I} \|  h_{k_i}(x_i) \|_2 \]
Noting that $\| h_{k_i}(x_i) \|_2 = 1$ for all $i$, and $\| \Delta h_{k_i}(x_i, y_i) \|_2 \leq 8 \sqrt{k_i} \cdot \sqrt{\eps}$
by Claim~\ref{cl:uv-hermite} above, we obtain

\begin{eqnarray*}
\|H_d(x) - H_d(y) \|_2 &\leq&   \sum_{\emptyset \neq I \subseteq [j]} \prod_{i \in I} \| \Delta h_{k_i}(x_i, y_i) \|_2 \\
&=& \sum_{i=1}^j \| \Delta h_{k_i}(x_i, y_i) \|_2 +
\sum_{I \subseteq [j], |I| \geq 2} \prod_{i \in I} \| \Delta h_{k_i}(x_i,y_i) \|_2 \\
&\leq&  8 \big(\sum_{i=1}^j \sqrt{k_i}\big) \cdot \sqrt{\eps} +  \sum_{|I|=2}^d \binom{d}{|I|} (8 \sqrt{d} \sqrt{\eps})^{|I|} \\
&\leq& 8 d \cdot \sqrt{\eps} + (1+8 \sqrt{d} \sqrt{\eps})^d - (1+8d^{3/2} \sqrt{\eps})\\
&\leq& O( d \cdot \sqrt{\eps})
\end{eqnarray*}
where the last inequality follows from the elementary bound $(1+8\sqrt{d}\sqrt{\eps})^d \leq 1+8d^{3/2}\sqrt{\eps}+O(d^3 \eps)$
and the fact that $\eps \leq 2^{-2d-1}$. This completes the proof of Claim~\ref{clm:mv-hermite}.

\subsubsection{Proof of Claim~\ref{cl:uv-hermite}}

We will need a crucial lemma:

\begin{lemma} \label{lemma:herm-eps}
For all $k \in [d]$ we have $$\|h_k (x) - h_k(x-\eps x) \|_2 \leq 3k \eps.$$
\end{lemma}

\begin{proof}
Note that $h_k(x-\eps x)$ is a degree-$k$ polynomial in $x$. Hence, by Taylor's theorem we deduce
\[
h_k(x)-h_k(x-\eps x) = - \sum_{i=1}^k h_k^{(i)}(x) (-\eps x)^i/i!   .
\]
The triangle inequality for norms now yields
\[
\|  h_k(x)-h_k(x-\eps x) \|_2 \leq \sum_{i=1}^k (\eps^{i}/i!) \cdot \| h_k^{(i)}(x) x^{i}\|_2 .
\]
It thus suffices to bound the term $\| h_k^{(i)}(x) x^{i} \|_2$. Recalling that $(h_k^{(i)}(x))^2 = i! \binom{k}{i}
(h_{k-i}(x))^2$ we have
\[
\E_x [(h_k^{(i)}(x))^2 x^{2i}] =  i! \binom{k}{i} \cdot \E_x [h_{k-i}^2(x) x^{2i}]  .
\]
For $i=1$, using the well-known relation
\[ \sqrt{k} h_k(x) + \sqrt{k-1}h_{k-2}(x) = xh_{k-1}(x) \]
and the orthonormality of the $h_i$'s, an easy calculation gives $\E_x[ h^2_{k-1}(x) x^2] = 2k-1$; hence,
$$\| h'_k(x) x \|_2
\leq \sqrt{2}k.$$ For $i>1$, by Cauchy-Schwartz we get
$$\E_x [h_{k-i}^2(x) x^{2i}] \leq \sqrt{\E_x [h_{k-i}^4(x)] \cdot \E_x [x^{4i}]}.$$
We now proceed to bound the RHS. By hypercontractivity, the first term can be bounded as follows
\[ \| h_{k-i} \|_4^2 \leq 3^{k-i} \| h_{k-i} \|_2^2 = 3^{k-i} .\]
For the second term we recall that, for $x \sim \N$, we have $\E_x [x^{4i}] =\frac{(4i)!}{2^{2i} (2i)!}$. Using the elementary
inequality $(2j)!/j! < 2^{2j} j!$ we conclude
\[ \E_x [h_{k-i}^2(x) x^{2i}] \leq 3^{k-i} \cdot 2^{i} \sqrt{(2i)!} \leq 3^k \cdot (4/3)^i \cdot i! \leq 4^k i! \]
hence, $$ \| h_k^{(i)}(x) x^{i}\|_2 \leq \sqrt{\binom{k}{i}} 2^k i! \leq 2^{3k/2} \cdot i! .$$

Therefore,
\[
\|h_k(x)-h_k(x-\eps x)\|_2 \leq \sqrt{2}k \cdot \eps + \eps \cdot 2^{3k/2} \cdot \sum_{j=1}^{k-1} \eps^{j}
           \leq 3k \cdot \eps
\]
where we used the fact $\eps \leq 2^{-2d} \leq 2^{-2k}$ which yields $\sum_{j=1}^{k-1} \eps^{j} \leq \sum_{j=1}^{\infty}
2^{-2kj} \leq 2^{-2k+1}$. The proof of the lemma is now complete.
\end{proof}

We now proceed to complete the proof of our claim. Let us write
\[
\Delta h_d(x,y) = h_d(x)-h_d(y) = q_1 (x) + q_2(x,z)
\]
where $q_1(x) = h_d(x) - h_d(x-\eps x)$ and $q_2(x,z) =h_d(x-\eps x) - h_d(x-\eps x + \beta z)$.

By the triangle inequality for norms it holds that
\[
\| \Delta h_d(x,y)\|_2 \leq \| q_1 \|_2 + \| q_2 \|_2
\]
hence it suffices to bound each of the terms in the RHS.

By Lemma~\ref{lemma:herm-eps} it follows that
\[ \| q_1 \|_2 \leq 3d \eps .\]
For the second term, we will show that
\[ \|q_2\|_2 \leq 5 \sqrt{d} \cdot \sqrt{\eps}.\]
Note that this suffices to complete the proof, since by our assumption on $\eps$, we have $d \cdot \eps<1$, which implies that
\[ \|\Delta h_d(x,y)\|_2 \leq 8 \sqrt{d} \sqrt{\eps}\]
as desired.

Now observe that $h_d(x-\eps x + \beta z)$ is a degree-$d$ polynomial in $x, z$. Let us denote $x' = (1-\eps)x$. By Taylor's
theorem we can write
\[ h_d(x'+\beta z) = h_d(x') + \sum_{i=1}^d (\beta^i/i!) h_d^{(i)}(x')z^i \]
or
\[
q_2(x, z) = - \sum_{i=1}^d (\beta^i/i!) h_d^{(i)}(x')z^i.
\]
By triangle inequality
\[
\|q_2 \|_2 \leq \sum_{i=1}^d (\beta^i/i!) \| h_d^{(i)}(x')  z^i \|_2
\]
For the terms in the RHS by independence we get
\[
\| h_d^{(i)}(x')  z^i \|_2 = \|h_d^{(i)}(x')\|_2  \cdot \|z^i \|_2
\]
For the second term above we have that $\|z^i \|_2 \leq 2^{i/2} \cdot \sqrt{i!}$.

Recalling that $h_d^{(i)}(x')^2 = i! \binom{d}{i} h_{d-i}^2(x')$ for the first term we have $$\|h_d^{(i)}(x')\|_2 = \sqrt{i!}
\sqrt{\binom{d}{i}} \cdot \|h_{d-i}(x')\|_2.$$ Since $x' = x-\eps x$ we apply Lemma~\ref{lemma:herm-eps} for $k=d-i$ and get
\[ \|h_{d-i}(x')\|_2 \leq \|h_{d-i}(x)\|_2 + 3(d-i)\eps \leq 2 \]
where the second inequality uses the assumption on the range of $\eps$.

Therefore,
\begin{eqnarray*}
\|q_2 \|_2 &\leq& \sum_{i=1}^d 2^{i/2+1} \sqrt{\binom{d}{i}} \beta^i \\
           &\leq& 4\sqrt{d \eps} + \beta \cdot \sum_{i=2}^d 2^{i/2+1} \sqrt{\binom{d}{i}} \beta^{i-1}\\
           &\leq& 4\sqrt{d \eps} + \sqrt{2\eps} \cdot \sum_{i=2}^d 2^{i/2+1} \sqrt{\binom{d}{i}} 2^{-d(i-1)} \\
           &\leq& 5\sqrt{d \eps}
\end{eqnarray*}

This completes the proof of Claim~\ref{cl:uv-hermite}.

\section{Boolean Average Sensitivity}
\label{sec:booleanas}

Let $\AS(n,d)$ denote the maximum possible average sensitivity of
any degree-$d$ PTF over $n$ Boolean variables. In this section we
prove the claimed bound in \theoremref{thm:boolas}:
\begin{equation}
\label{eqn:booleanas} \AS(n,d) \leq 2^{O(d)} \cdot \log n \cdot n^{1-1/(4d+2)}
\end{equation}

For $d=1$ (linear threshold functions) it is well known that
$\AS(n,1)=2^{-n}{n\choose n/2} = \Theta(\sqrt{n})$.  Also, notice
that the RHS of (\ref{eqn:booleanas}) is larger than $n$ for
$d=\omega(\sqrt{\log n})$, yielding a trivial bound of $\AS(n,d)\leq
n$.  Therefore throughout this section we shall assume $d$ satisfies
$2 \leq d \leq O(\sqrt{\log n})$.

\subsection{Overview of proof}
\label{sec:booloverview}

The high-level approach to proving \theoremref{thm:boolas} is a
combination of a case analysis and a recursive bound.

For certain types of PTFs (``$\tau$-regular'' PTFs; see Section \ref{sec:regprelim} for a precise definition) we argue directly
that the average sensitivity is small, using arguments similar to the Gaussian case together with the invariance principle.  In
particular, we show:

\begin{claim}
\label{claim:regularas} Suppose $f=\sgn(p)$ is a $\tau$-regular degree-$d$ PTF where $\tau\eqdef n^{-(4d+1)/(4d+2)}$. Then,
$$\AS(f) \leq O(d\cdot n^{1-1/(4d+2)})$$
\end{claim}
\claimref{claim:regularas} follows directly from Lemma \ref{lemma:regular}, which we prove in \sectionref{sec:regular}.

For PTFs that are not $\tau$-regular, we show that there is a
not-too-large value of $k$ (at most $K\eqdef 2d\log n/\tau$), and a
collection of $k$ variables (the variables whose influence in $p$
are largest), such that the following holds: if we consider all
$2^k$ subfunctions of $f$ obtained by fixing the variables in all
possible ways, a ``large'' (at least $1/2^{O(d)}$) fraction of the
restricted functions have low average sensitivity. More precisely,
we show:

\begin{claim} \label{claim:irregas} Let $K \eqdef 2d\log n/\tau$ where
  $\tau\eqdef n^{-(4d+1)/(4d+2)}$. Suppose $f=\sgn(p)$ is a
  degree-$d$ PTF that is not $\tau$-regular. Then for some $1 \leq k
  \leq K,$ there is a set of $k$ variables with the following
  property: for at least a $1/2^{O(d)}$ fraction of all $2^k$
  assignments $\rho$ to those $k$ variables, we have $$\AS(f_\rho)
  \leq O(d\cdot (\log n)^{1/4}\cdot n^{1-1/(4d+2)})$$
\end{claim}

The proof of \claimref{claim:irregas} is given in Section \ref{sec:pfofirreg}. We do this by generalizing the ``critical
index'' case analysis from \cite{Servedio:07cc}.  We define a notion of the $\tau$-critical index of a degree-$d$ polynomial; a
$\tau$-regular polynomial $p$ is one for which the $\tau$-critical index is 0. If the $\tau$-critical index of $p$ is some
value $k \leq 2d\log n/\tau$, we restrict the $k$ largest-influence variables (see Section \ref{sec:small}).  If the
$\tau$-critical index is larger than $2d\log n/\tau$, we restrict the $k=2d\log n/\tau$ largest-influence variables in $p$ (see
\sectionref{sec:large}).

\subsubsection{Proof of main result (\theoremref{thm:boolas})
  assuming \claimref{claim:regularas} and \claimref{claim:irregas}}
\label{sec:recurse}

Given these two claims it is not difficult to obtain the final result. In \claimref{claim:irregas}, we note that the $k$
restricted variables may each contribute at most $1$ to the average sensitivity of $f$ (recall that average sensitivity is
equal to the sum of influences of each variable), and that the total influence of the remaining variables on $f$ is equal to
the expected average sensitivity of $f_\rho$, where the expectation is taken over all $2^k$ restrictions $\rho$.  Since each
function $f_{\rho}$ is itself a degree-$d$ PTF over at most $n$ variables, we have the following recursive constraint on
$\AS(n,d)$:
\begin{eqnarray*}
  \AS(n,d) \leq \max\{&&O(d\cdot n^{1-1/(4d+2)}),\\
  &&\max_{1 \leq k \leq K, \ \ 1/2^{O(d)} \leq \alpha \leq 1} \{k +
  \alpha \cdot O(d\cdot (\log n)^{1/4}\cdot n^{1-1/(4d+2)}) +
  (1-\alpha)\AS(n,d)\}\}.
\end{eqnarray*}
It is easy to see that the maximum possible value of $\AS(n,d)$ subject to the above constraint is at most the maximum possible
value of $\AS'(n,d)$ that satisfies the following weaker constraint:
\[
\AS'(n,d) \leq K + \left(1-\frac{1} {2^{O(d)}}\right)\AS'(n,d)
\]
which is satisfied by $\AS'(n,d) \leq 2^{O(d)} \cdot \log n \cdot n^{1-1/(4d+2)}$.

\subsection{Regularity and the critical index of polynomials}
\label{sec:regprelim}

In \cite{Servedio:07cc} a notion of the ``critical index'' of a linear form was defined and subsequently used in
\cite{OdonnellServedio:08,DiakonikolasServedio:09,DGJ+09}. We now give a generalization of the critical index notion for
polynomials.

\begin{definition}
Let $p: \bits^n \to \R$ and $\tau>0$. Assume the variables are ordered such that $\Inf_i(f)\geq \Inf_{i+1}(f)$ for all $i \in
[n-1]$.  The {\em $\tau$-critical index} of $f$ is the least $i$ such that:
\begin{equation}
\frac{\Inf_{i+1}(p)}{\sum_{j=i+1}^n\Inf_j(p)}\leq \tau. \label{eq:reg}
\end{equation}
If (\ref{eq:reg}) does not hold for any $i$ we say that the $\tau$-critical index of $p$ is $+ \infty.$ If $p$ is has
$\tau$-critical index 0, we say that $p$ is {\em $\tau$-regular}.
\end{definition}

The following simple lemma will be useful for us. It says that the total influence $\sum_{i=j+1}^n \Inf_i (p)$ goes down
exponentially as a function of $j$ prior to the critical index:

\begin{lemma} \label{lemma:expshrink}
Let $p: \bits^n \to \R$ and $\tau>0$. Let $k$ be the $\tau$-critical index of $p$. For $0\leq j \leq k$ we have
\[
\sum_{i=j+1}^n \Inf_i (p) \leq (1 - \tau)^j \cdot \Inf(p).
\]
\end{lemma}

\begin{proof}
The lemma trivially holds for $j=0.$  In general, since $j$ is at most $k$, we have that
\[
\Inf_j(p) \geq \tau \cdot \sum_{i=j}^n \Inf_i(p),
\]
or equivalently
\[
\sum_{i=j+1}^n \Inf_{i}(p) \leq (1 - \tau) \cdot \sum_{i=j}^n \Inf_{i}(p)\] which yields the claimed bound.
\end{proof}

Let $p: \bits^n \to \R$ be a degree-$d$ polynomial. We note here that the total influence of $p$ is within a factor of $d$ of
the sum of squares of the non-constant coefficients of $p$:
\[ \sum_{S \neq \emptyset} \widehat{p}(S)^2 \leq \sum_{i=1}^n\sum_{S\ni i}
\widehat{p}(S)^2 =\sum_{i=1}^n \Inf_i(p)= \sum_{S\subseteq  [n]} |S|\cdot \widehat{p}(S)^2 \leq d\sum_{S \neq
\emptyset}\widehat{p}(S)^2,
\] where the final inequality holds since $\widehat{p}(S) \neq 0$
only for sets $|S| \leq d.$

\subsection{Restrictions and the influences of variables in polynomials}

Let $p: \bits^n \to \R$ be a degree-$d$ polynomial. The goal of this section is to understand what happens to the influences of
a variable $x_\ell$, $\ell > k$, when we do a random restriction to variables $x_1,\dots,x_k.$

We start with the following elementary claim:
\begin{claim}
Let $\rho$ be a randomly chosen assignment to the variables $x_1,\dots,x_k$.  Fix any $S \subseteq \{k+1,\dots,n\}$. Then for
any polynomial $p: \bits^n \to \R$ we have

\[
\widehat{p_\rho}(S) = \sum_{T \subseteq [k]} \widehat{p}(S \cup T)\rho_T,\]

and so we have
\begin{equation}
\E_\rho[\widehat{p_\rho}(S)^2] = \sum_{T \subseteq [k]} \widehat{p}(S \cup T)^2. \label{eq:eprhosquare}
\end{equation}
\end{claim}
In words, all the Fourier weight on sets of the form $S \,\cup \{$some restricted variables$\}$ ``collapses'' down onto $S$ in
expectation. A corollary of this is that in expectation, the influence of an unrestricted variable $x_\ell$ does not change
when we do a restriction:

\begin{corollary} \label{cor:expinf}
Let $\rho$ be a randomly chosen assignment to the variables $x_1,\dots,x_k$.  Fix any $\ell \in \{k+1,\dots,n\}$. Then for any
polynomial $p: \bits^n \to \R$ we have
\[
\E_\rho[\Inf_\ell(p_\rho)] = \Inf_\ell(p).
\]
\end{corollary}

\begin{proof}
\begin{eqnarray*}
\E_\rho[\Inf_\ell(p_\rho)] &=& \E_\rho\left[\sum_{\ell \in S \subseteq
\{k+1,\dots,n\}} \widehat{p_\rho}(S)^2\right]\\
&=& \sum_{T \subseteq [k]} \sum_{\ell \in S \subseteq
\{k+1,\dots,n\}} \widehat{p}(S \cup T)^2\\
&=& \sum_{U \ni \ell} \widehat{p}(U)^2 = \Inf_\ell(p).
\end{eqnarray*}
\end{proof}

\subsubsection{Influences of low-degree polynomials behave nicely under
  restrictions}

In this subsection we prove the following lemma: For a low-degree polynomial, a random restriction with very high probability
does not cause any variable's influence to increase by more than a polylog$(n)$ factor.

\begin{lemma}\label{lem:smalldevinf}
  Let $p(x_1,\dots,x_n)$ be a degree-$d$ polynomial. Let $\rho$ be a
  randomly chosen assignment to the variables $x_1,\dots,x_k$.  Fix
  any $t>e^{2d}$ and any $\ell \in [k+1,n]$. With probability at least
  $1-\exp(-\Omega(t^{1/d}))$ over the choice of $\rho$, we have
$$\Inf_\ell(p_\rho)\leq t\cdot 3^d  \Inf_\ell(p).$$
In particular, for $t=\log^dn$, we have that with probability at least $1-n^{-\omega(1)}$, every variable $\ell\in [k+1,n]$ has
$\Inf_\ell(p_\rho)\leq (3\log n)^d \cdot \Inf_\ell(p)$. \ignore{\lnote{This is
  Lemma 8 of the PTF regularity lemma paper (Section 4.1). Previously
  we had a weaker bound of $\Inf_\ell(p_\rho)\leq \log^{2d}(n)
  \cdot\Inf_\ell(p)$.}}
\end{lemma}

\begin{proof}
  Since $\Inf_\ell(p_\rho)$ is a degree-$2d$ polynomial in $\rho$,
  \lemmaref{lem:smalldevinf} follows as an immediate consequence of
  \theoremref{thm:deg-d-chernoff} if we can upper bound
  $||\Inf_\ell(p_\rho)||_2.$ We use the bound in
  \lemmaref{lem:inf2norm}, stated and proven below.
\end{proof}
\begin{lemma}\label{lem:inf2norm}
  Let $p(x_1,\dots,x_n)$ be a degree-$d$ polynomial. Let $\rho$ be a
  randomly chosen assignment to the variables $x_1,\dots,x_k$, and let
  $\ell \in [k+1,n]$. Then $\Inf_\ell(p_\rho)$ is a degree-$2d$
  polynomial in variables $\rho_1,\dots,\rho_k$, and
$$||\Inf_\ell(p_\rho)||_2 \leq 3^d\cdot \Inf_\ell(p).$$
\end{lemma}
\begin{proof}
The triangle inequality tells us that we may bound the $2$-norm of each squared-coefficient separately:
$$||\Inf_\ell(p_\rho)||_2 \leq \sum_{\ell\in S \subseteq [k+1,n]}
||\widehat{p}_\rho(S)^2||_2.$$ Since $\widehat{p}_\rho(S)$ is a degree-$d$ polynomial, Bonami-Beckner (i.e.,
$(4,2)$-hypercontractivity) tells us that
$$||\widehat{p}_\rho(S)^2||_2 = || \widehat{p}_\rho(S) ||_4^2 \leq 3^d
||\widehat{p}_\rho(S)||_2^2,$$ hence
$$||\Inf_\ell(p_\rho)||_2 \leq 3^d \sum_{\ell\in S \subseteq [k+1,n]}
||\widehat{p}_\rho(S)||_2^2 = 3^d\cdot \Inf_\ell(p)$$ where the last equality is by \corollaryref{cor:expinf}.
\end{proof}

\subsection{The regular case}
\label{sec:regular}

In this section we prove that regular degree-$d$ PTF's have low average sensitivity.  In particular, we show:

\begin{lemma}
\label{lemma:regular} Fix $\tau=n^{-\Theta(1)}$. Let $f$ be a $\tau$-regular degree-$d$ PTF. Then,
\[ \AS(f) \leq O(d \cdot n \cdot \tau^{1/(4d+1)})\]
\end{lemma}

\claimref{claim:regularas} follows directly from the above lemma, recalling we choose $\tau \eqdef n^{-(4d+1)/(4d+2)}$.
However, the lemma will also be useful in the ``small critical index'' case for a slightly larger regularity parameter $\tau$.

\begin{proof}
  Let $f : \bits^n \to \R$ be a degree-$d$ PTF, i.e. $f=\sign(p)$
  where $p$ is $\tau$-regular.  We may assume that $p$ is normalized
  such that $\sum_{0 < |S|\leq d}\widehat{p}(S)^2=1$.

  First we note that flipping the $i$-th bit of an input $x\in\bn$
  changes the value of $p$ by the magnitude of its partial derivative
  with respect to $i$:
  \[ 2D_ip(x) = 2\sum_{S\ni i}\widehat{p}(S)x_{S-\{i\}} \] It follows
  that:
  \[ \Inf_i(f) \leq \Pr_{x\in\bn}[|p(x)|\leq|2D_ip(x)|] \]

  Therefore, bounding from above the influence of variable $i$ in $f$
  can be done by showing the following:

\begin{enumerate}
\item $p(x)$ has small magnitude, $|p(x)|\leq t$ for some threshold
  $t$, with small probability.
\item $2D_ip(x)$ has large magnitude, $|2D_ip(x)| \geq t$, with small
  probability.
\end{enumerate}

We bound the probability of the first event using the anti-concentration property of regular low-degree polynomials, as implied
by the invariance principle along with \theoremref{thm:carberywright}. For the second event we use the tail bound for
degree-$d$ polynomials (\theoremref{thm:deg-d-chernoff}).

We will take our threshold $t$ to be $t \eqdef \tau^{1/4}$, where $\tau$ is the regularity parameter of $p$.

\subsubsection{Bounding the probability of the first event}

By the $\tau$-regularity of $p$, for all $i \in [n]$ we have $\Inf_i(p) \leq \tau \cdot \Inf (p) \leq d \cdot \tau$ where the
last inequality follows by the assumed normalization. With this bound, the invariance principle (\theoremref{thm:invariance})
tells us that $\Pr_{x\in\bn}[|p(x)|\leq \tau^{1/4}]$ differs from
$\Pr_{\mathcal{G}_1,\ldots,\mathcal{G}_n}[|p(\mathcal{G})|\leq\tau^{1/4}]$ by at most $O(d\cdot (d\tau)^{1/(4d+1)}) = O(d\cdot
\tau^{1/(4d+1)})$. Applying the anti-concentration bound of Carbery and Wright for polynomials in Gaussian random variables
(\theoremref{thm:carberywright}), we get:
\begin{eqnarray*}
  \Pr_x[|p(x)|\leq \tau^{1/4}] &\leq&
  \Pr_{\mathcal{G}_1,\ldots,\mathcal{G}_n}[|p(\mathcal{G})|\leq\tau^{1/4}]
  + O(d \tau^{1/(4d+1)}) \\
  &\leq& O(d\cdot \tau^{1/4d}) + O(d \cdot \tau^{1/(4d+1)}) \\
  &=& O(d\cdot \tau^{1/(4d+1)}).
\end{eqnarray*}

\subsubsection{Bounding the probability of the second event}

Next we consider $\Pr_x[|2D_ip(x)|\geq\tau^{1/4}]$. Note that $2D_ip$ is a degree-$(d-1)$ polynomial whose $l_2$ norm is small:
\[ \| 2D_i p \| = 2 \sqrt{\sum_{S\ni i}\widehat{p}(S)^2} = 2 \sqrt{
  \Inf_i(p)} \leq 2\sqrt{d\cdot\tau}.  \] By
(\theoremref{thm:deg-d-chernoff}), we get that
\begin{eqnarray*}
\Pr_x[|2D_ip(x)|\geq \tau^{1/4}] &\leq& \Pr_x[|2D_ip(x)|\geq
\tau^{-1/4}/(2\sqrt{d}) \cdot \| 2D_i p \| ]\\
&\leq& \exp(-\tau^{-1/(2d)}/(2\sqrt{d})^{2/d}) = \exp(-\Theta(1)\cdot\tau^{-1/(2d)}) \ll O(d\cdot \tau^{1/(4d+1)}).
\end{eqnarray*}
(In the second inequality, we were able to apply the concentration bound since, by our assumptions on $d$ and $\tau$, we indeed
have that $\tau^{-1/4}/(2 \sqrt{d}) > e^d.$)

Hence, we have shown that:
\begin{eqnarray*}
  \Inf_i(f) &\leq& \Pr_{x\in\bn}[|p(x)|\leq|2D_ip(x)|] \\
  &\leq& \Pr_x[|p(x)|\leq\tau^{1/4}]  +  \Pr_x[|2D_ip(x)|\geq
  \tau^{1/4}] \\
&=& O(d \cdot \tau^{1/(4d+1)}).
\end{eqnarray*}

Since this holds for all indices $i \in [n]$, we have the following bound on the average sensitivity of $f=\sgn(p)$:
\[ \AS(f) \leq O(d \cdot n \cdot \tau^{1/(4d+1)}). \]

\end{proof}

\subsection{The small critical index case} \label{sec:small}

Let $f=\sign(p)$ be such that the $\tau$-critical index of $p$ is some value $k$ between $1$ and $K= 2d\log n/\tau$. By
definition, the sequence of influences $\Inf_{k+1}(p),\dots,\Inf_n(p)$ is $\tau$-regular. We essentially reduce this case to
the regular case for a regularity parameter $\tau'$ somewhat larger than $\tau$.

Consider a random restriction $\rho$ of all the variables up to the critical index.  We will show the following:

\begin{lemma}
For a $1/2^{O(d)}$ fraction of restrictions $\rho$, the sequence of influences $\Inf_{k+1}(p_\rho),$ $\dots,$ $\Inf_n(p_\rho)$
is $\tau'$-regular, where $\tau' \eqdef (3\log n)^d\cdot \tau$.\ignore{ \lnote{Lemma 7 of the PTF regularity lemma paper gives
  $\tau'=\tau\cdot O(d\ln\littlefrac{1}{\tau})^d$. Given our choice of
  $\tau = \poly(\littlefrac{1}{n})$ both lemmas yield essentially the
  same $\tau'$.} }
\end{lemma}

By our choice of $\tau = n^{-(4d+1)/(4d+2)}$, we have that $\tau' = n^{-\Theta(1)}$, and so we may apply
\lemmaref{lemma:regular} to these restrictions to conclude that the associated PTFs have average sensitivity at most $O(d \cdot
n \cdot ({\tau'})^{1/(4d+1)})$.

\begin{proof}

Since the sequence of influences $\Inf_{k+1}(p),\dots,\Inf_n(p)$ is $\tau$-regular, we have
\[  \frac{\Inf_{i}(p)}{\sum_{j=k+1}^n \Inf_{j}(p)} \leq \tau  \]
for all $i \in [k+1, n]$.

We want to prove that for a $1/2^{O(d)}$ fraction of all $2^k$ restrictions $\rho$ to $x_1, \ldots, x_k$ we have

\[  \frac{\Inf_{i}(p_{\rho})}{\sum_{j=k+1}^n \Inf_{j}(p_{\rho})} \leq \tau'  \]

for all $i \in [k+1, n]$.

To do this we proceed as follows: \lemmaref{lem:smalldevinf} implies that, with very high probability over the random
restrictions, we have $\Inf_{i}(p_{\rho}) \leq (3\log n)^d\cdot \Inf_{i}(p)$, for all $i \in [k+1, n]$. We need to show that
for a $1/2^{O(d)}$ fraction of all restrictions the denominator of the fraction above is at least $\sum_{j=k+1}^n \Inf_{j}(p)$
(its expected value). The lemma then follows by a union bound.

We consider the degree-$2d$ polynomial $A(\rho_1,\dots,\rho_k) \eqdef \sum_{j=k+1}^n \Inf_{j}(p_{\rho})$ in variables
$\rho_1,\dots,\rho_k.$ The expected value of $A$ is $\E_\rho[A] = \sum_{j=k+1}^n \Inf_{j}(p) = \widehat{A}(\emptyset)$. We
apply the \theoremref{thm:aushas} for $B = A - \widehat{A}(\emptyset)$. We thus get $\Pr_{\rho} [B > 0 ] > 1/2^{O(d)}$.  We
thus get $\Pr_{\rho} [A > \E_\rho[A] ] > 1/2^{O(d)}$ and we are done.
\end{proof}

\subsection{The large critical index case}
\label{sec:large}

Finally we consider PTFs $f=\sign(p)$ with $\tau$-critical index greater than $K= 2d\log n/\tau$.  Let $\rho$ be a restriction
of the first $K$ variables $\mathcal{H} = \{1,\dots,K\}$; we call these the ``head'' variables. We will show the following:

\begin{lemma}
For a $1/2^{O(d)}$ fraction of restrictions $\rho$, the function $\sign(p_\rho(x))$ is a constant function.
\end{lemma}

\begin{proof}
By \lemmaref{lemma:expshrink}, the surviving variables $x_{K+1},\ldots,x_n$ have very small total influence in $p$:
\begin{equation}
\sum_{i=K+1}^n\Inf_i(p) = \sum_{i=K+1}^n\sum_{S\ni i}\widehat{p}(S)^2 \leq (1-\tau)^{K} \cdot \Inf(p) \leq d/n^{2d}.
\label{eq:hi}
\end{equation} Therefore, if we let $p'$ be the truncation of $p$
comprising only the monomials with all variables in $\mathcal{H}$,
\[ p'(x_1,\ldots,x_k) = \sum_{S\subset\mathcal{H}} \widehat{p}(S)x_S \] we
know that almost all of the original Fourier weight of $p$ is on the coefficients of $p'$:
\[ 1 \geq \mathop{\sum_{S\subset\mathcal{H}}}_{|S|>0}\widehat{p}(S)^2 \geq
1 - \sum_{i=K+1}^n \Inf_i(p) \geq 1 - d/n^{2d}\]

We now apply \theoremref{thm:aushas} to $p'$ \footnote{after a very slight
  rescaling so the non-constant Fourier coefficients of $p'$ have sum
  of squares equal to 1; this does not affect the bound we get because of the big-O.} and get:
\[ \Pr_{x\in\{-1,1\}^K}[|p'(x)|\geq 1/2^{O(d)}] \geq 1/2^{O(d)}. \]

In words, for a $1/2^{O(d)}$ fraction of all restrictions $\rho$ to $x_1,\dots,x_K$, the value $p'(\rho)$ has magnitude at
least $1/2^{O(d)}$.

For any such restriction, if the function $f_\rho(x) = \sign(p_\rho(x_{K+1},\dots,x_n))$ is not a constant function it must
necessarily be the case that:
\[
\sum_{0 < |S|\subseteq \{x_{K+1},\dots,x_n\}} |\widehat{p_\rho}(S)| \geq 1/2^{O(d)}
\]

As noted in (\ref{eq:hi}), each tail variable $\ell > K$ has very small influence in $p$:
\[ \Inf_\ell(p) \leq \sum_{i=K+1}^n\Inf_i(p) =
d/n^{2d}\]

Applying \lemmaref{lem:smalldevinf}, we get that for the overwhelming majority of the $1/2^{O(d)}$ fraction of restrictions
mentioned above, the influence of $\ell$ in $p_\rho$  is not much larger than the influence of $\ell$ in $p$:

\begin{equation} \Inf_\ell(p_\rho) \leq (3\log n)^{d}\cdot
  \Inf_\ell(p) \leq d\cdot (3\log n)^{d}/n^{2d} \label{eq:reallylittle}
\end{equation}

Using Cauchy-Schwarz, we have

\begin{eqnarray*}
\sum_{S \ni \ell, S \subseteq \{x_{K+1},\dots,x_n\}} |\widehat{p_\rho}(S)| &\leq& n^{d/2} \cdot \sqrt{\sum_{S \ni \ell, S
\subseteq \{x_{K+1},\dots,x_n\}} \widehat{p_\rho}(S)^2}\\
&=& n^{d/2}
\sqrt{\Inf_\ell(p_\rho)}\\
&\leq&n^{-\Omega(1)}
\end{eqnarray*}
where we have used (\ref{eq:reallylittle}) (and our upper bound on $d$). From this we easily get that

\[
\sum_{0 < |S|\subseteq \{x_{K+1},\dots,x_n\}} |\widehat{p_\rho}(S)| \leq n^{- \Omega(1)} \ll 1/2^{O(d)}
\]

We have established that for a $1/2^{O(d)}$ fraction of all restrictions to $x_1,\dots,x_K$, the function $f_\rho =
\sign(p_\rho)$ is a constant function, and the lemma is proved.
\end{proof}

\subsection{Proof of \claimref{claim:irregas}}
\label{sec:pfofirreg}

If $f$ is a degree-$d$ PTF that is not $\tau$-regular, then its $\tau$-critical index is either in the range $\{1,\dots,K\}$ or
it is greater than $K.$

In the first case (small critical index case), as shown in
\sectionref{sec:small}, we have that for a $1/2^{O(d)}$ fraction of
restrictions $\rho$ to variables $x_1,\dots,x_k$, the total influence of $f_\rho = \sign(p_\rho)$ is at most
\[
O(d \cdot n \cdot (\tau')^{1/(4d+1)}) = O(d\cdot (\log n)^{1/4} \cdot n^{1-1/(4d+2)}),\] so the conclusion of
\claimref{claim:irregas} holds in this case.

In the second case (large critical index case), as shown in
\sectionref{sec:large}, for a $1/2^{O(d)}$ fraction of restrictions
$\rho$ to $x_1,\dots,x_K$ the function $f_\rho$ is constant and hence has zero influence, so the conclusion of
\claimref{claim:irregas} certainly holds in this case as well. \qed

\section{A Fourier-Analytic Bound on Boolean Average Sensitivity} 
\label{sec:booleanas2}

In this section, we present a simple proof of the following upper bound on the
average sensitivity of a degree-$d$ PTF (\theoremref{thm:boolas2}):
$$\AS(n,d) \leq 2 n^{1-1/2^d} \mper$$

We recall here the definition of the formal derivative of a function
$f : \bn \to \R$.
\[
D_i p(x) = \sum_{S \ni i} \widehat{p}_S x_{S - \{i\}}.
\]
It is easy to see that,
\begin{equation}\label{eq:derivative}
D_i p(x) = {\frac 1 2} x_i[p(x) - p(x^{\oplus i})] = \frac{1}{2}
\left(\frac{p(x)-p(x^{\oplus i})}{x_i}\right)
\end{equation}
where ``$x^{\oplus i}$'' means ``$x$ with the $i$-th bit flipped.''

For a Boolean function $f$, we have $D_i f(x) = \pm 1$ iff flipping
the $i$th bit flips $f$; otherwise $D_i f(x)=0.$  So we have
\[
\Inf_i(f) = \E[|D_i f(x)|].
\]

\begin{lemma} \label{lem:2func} Fix $i \neq j \in [n]$.  Let $f,g :
\bn \to \R$ be functions such that $f$ is independent of the
$i$\th bit $x_i$ and
$g$ is independent of the $j$\th bit $x_j.$  Then
\[
\E_x[x_i x_j f(x) g(x)] \leq {\frac {\Inf_i(g)+\Inf_j(f)}2}.
\]
\end{lemma}
\begin{proof}
First, note that the influence of $i$\th coordinate on a function $f$
can be written as:
\begin{align}\label{eq:influence}
\Inf_i(f)  = \E_{x_{-i}}[\Var_{x_i}[f(x)]] = \E_{x}\left[ \left(\frac{|f(x^{\oplus i}) -
f(x)|}{2} \right)^2 \right]  = \E_{x_{-i}}\left[\left|\E_{x_i}[x_i
f(x)]\right|^2\right]
\end{align}

As $f$ is independent of $x_i$ and $g$ is independent of $x_j$, we can
write,
\begin{align*}
  \E_x[x_i x_j& f(x) g(x)] = \E_{x_{-\{i,j\}}} \E_{x_i,x_j}\left[x_i
    x_j f(x)g(x)\right] \\
  &= \E_{x_{-\{i,j\}}} \left[\E_{x_i}[x_i
    g(x)]\E_{x_j}[x_j f(x)]\right] \\
  &\leq \E_{x_{-\{i,j\}}} \left[ \frac{1}{2}|\E_{x_i}[x_i g(x)]|^2 +
    \frac{1}{2}|\E_{x_j}[x_j f(x)]|^2\right] \qquad
  (\text{using } ab \leq \littlefrac{1}{2}(a^2 + b^2))\\
  & \leq \frac{\Inf_j(f)+\Inf_i(g)}{2} \qquad \qquad \qquad
  \qquad\qquad \qquad \qquad (\text{using \equationref{eq:influence}})
\end{align*}
\end{proof}


\theoremref{thm:boolas2} is shown using an inductive argument over the degree
$d$.  Central to this inductive argument is the following lemma
relating the influences of a degree-$d$ PTF $\sgn(p(x))$ to the
degree-$(d-1)$ PTFs obtained by taking formal derivatives of $p$.
\begin{lemma} \label{lem:key}  For a PTF $f = \sgn(p(x))$ on $n$
variables and $i \in [n]$,
$\Inf_i(f) = \E[f(x) x_i \sgn(D_i p(x))].$
\end{lemma}
The following simple claim will be useful in the proof of the above
lemma.
\begin{claim}\label{claim:sgn}
For two real numbers $a,b$,  if $\sgn(a) \neq \sgn(b)$ then
$$ \sgn(\sgn(a)-\sgn(b)) = \sgn(a-b)$$
\end{claim}
\begin{proof}
If $\sgn(a) = 1$ and $\sgn(b) = -1$ ($a \geq 0$, $b < 0$) then $a - b
\geq 0$.  Hence in this case, $\sgn(a-b) = 1 = \sgn(1-(-1)) =
\sgn(\sgn(a)-\sgn(b))$.  On the other hand, if $\sgn(a) = -1$ and
$\sgn(b) = 1$, then $\sgn(a-b) = -1 = \sgn((-1)-1) =
\sgn(\sgn(a)-\sgn(b))$.
\end{proof}
%
\begin{proof}[of \lemmaref{lem:key}]
The influence of the $i$\th coordinate is given by,
\begin{align}\label{eq:inflexpand}
\Inf_i (f)&=\E\left[{\frac 1 2}|f(x) - f(x^{\oplus i})|\right]
\nonumber\\
      &=\E\left[{\frac 1 2}\left(f(x) - f(x^{\oplus i})\right)\sgn\left(f(x)-f(x^{\oplus i})\right)\right]
\end{align}
Consider an $x$ for which $f(x) \neq f(x^{\oplus i})$.  In this case,
we can use \claimref{claim:sgn} to conclude:
\begin{align*}
 \sgn\left(f(x)-f(x^{\oplus i})\right) & = \sgn\left(p(x) - p(x^{\oplus i})\right) \mcom\\
& = \sgn(2x_i D_i p(x)) = x_i\sgn(D_i p(x))\mper \quad  (\text{using
\eqref{eq:derivative}})
\end{align*}
Hence for an $x$ with $f(x) \neq f(x^{\oplus i})$,
$$\left(f(x) - f(x^{\oplus
i})\right)\sgn\left(f(x)-f(x^{\oplus i})\right) = \left(f(x) - f(x^{\oplus
i})\right) x_i \sgn(D_i p(x)) \mper$$
On the other hand, if $f(x) = f(x^{\oplus i})$ then the above equation
continues holds since both the sides evaluate to $0$.
Substituting this equality into \equationref{eq:inflexpand} yields,
\begin{align*}
\Inf_i(f) = {\frac 1 2} \E\left[f(x) x_i \sgn(D_i p(x))\right] -
\frac{1}{2}\E\left[f(x^{\oplus i}) x_i \sgn(D_i p(x))  \right]\mper
\end{align*}
Notice that the $i$\th coordinate $(x^{\oplus i})_i$ of $x^{\oplus i}$ is given by
$-x_i$.  Since $D_i p$ is independent of the $i$\th
coordinate $x_i$, we have $D_i p(x) = D_i p(x^{\oplus i})$.  Rewriting
the above equation, we get
\begin{align*}
\Inf_i(f) &= {\frac 1 2} \E\left[f(x) x_i \sgn(D_i p(x))\right] +
\frac{1}{2}\E\left[f(x^{\oplus i}) (x^{\oplus i})_i \sgn(D_i
p(x^{\oplus i}))  \right] \mcom\\
&=  \E\left[f(x) x_i \sgn(D_i p(x))\right]  \qquad ((x^{\oplus i})\text{ is also uniformly distributed})
\end{align*}
\end{proof}

\begin{theorem}  \label{thm:recursiverelation}
Let $\AS(n,d)$ denote the max possible average sensitivity
of any degree-$d$ PTF on $n$ variables.  Then we have
\[ \AS(n,d) \leq \sqrt{n + n \cdot \AS(n,d-1)}.\]
\end{theorem}
\begin{proof}
\begin{eqnarray}
\Inf(f) &=& \sum_i \Inf_i(f) \nonumber\\
& = &  \sum_{i} \E[f(x) x_i \sgn(D_i p(x))] \quad \quad
\text{(by~\lemmaref{lem:key})}\nonumber \\
&=& \E[f(x) \sum_{i} x_i \sgn(D_i p(x))]  \nonumber \\
&\leq&
\sqrt{\E[f(x)^2]}
\cdot \sqrt{ \E[(\sum_{i} x_i \sgn(D_i p(x)))^2]} \label{eq:cs1}\\
&=& 1 \cdot\sqrt{ \E[\sum_{i,j} x_i x_j \sgn(D_i p(x))\sgn(D_j p(x))]}
\label{eq:exp}\\
&\leq& \sqrt{ \E[\sum_{i} x_i^2 \sgn(D_i p(x))^2] + \sum_{i \neq j}
\Inf_i(\sgn(D_j p(x)))} \label{eq:use2func}\\
&=& \sqrt{n +\sum_{i \neq j} \Inf_i(\sgn(D_j p(x)))}.
\label{eq:end}
\end{eqnarray}
Here (\ref{eq:cs1}) is the Cauchy-Schwarz inequality, (\ref{eq:exp}) is
expanding the square. Step (\ref{eq:use2func}) uses
\lemmaref{lem:2func} which we may apply since $D_i p(x)$ does not
depend on $x_i$.

Observe that for any fixed $j'$, we have $D_{j'} p(x)$ is a degree-$(d-1)$
polynomial and $\sgn(D_{j'} p(x))$ is a degree-$(d-1)$ PTF.  Hence, by
definition we have,
$$\sum_{i \neq {j'}} \Inf(\sgn(D_{j'} p(x))) \leq \AS(n,d-1) \mcom$$
for all $j' \in [n]$.  Therefore the quantity $\sum_{i \neq
j} \Inf(\sgn(D_j p(x))) \leq n \cdot \AS(n,d-1)$, finishing the proof.
\end{proof}

The bound on average sensitivity (\theoremref{thm:boolas2}) follows
immediately from the above recursive relation.

\begin{proof}[of \theoremref{thm:boolas2}]
  Clearly, we have $\AS(n,0) = 0$.  For $d = 1$,
  \theoremref{thm:recursiverelation} yields $\AS(n,1) \leq \sqrt{n}$.
  Now suppose $\AS(n,d) = 2n^{1-1/2^d}$ for $d \geq 1$, then by
  \theoremref{thm:recursiverelation},
$$ \AS(n,d+1) \leq \sqrt{n + n \cdot \AS(n,d)} \leq \sqrt{4
n^{2-1/2^{d}}} = 2n^{1-1/2^{d+1}}\mcom$$ finishing the proof.
\end{proof}

\section{Boolean average sensitivity vs noise sensitivity}
\label{sec:booleanns}

Our results on Boolean noise sensitivity are obtained via the following simple reduction which translates any upper bound on
average sensitivity for degree-$d$ PTFs over Boolean variables into a corresponding upper bound on noise sensitivity. This
theorem is inspired by the proof of noise sensitivity of halfspaces by Peres \cite{Peres:04}.

\begin{theorem} \label{thm:reduction}Let $\NS(\eps,d)$ denote the
  maximum noise sensitivity of a degree $d$-PTF at a noise rate of
  $\eps$.  For all $0 \leq \eps \leq 1$ if $m = \lfloor \frac{1}{\eps}
  \rfloor$ then,
$$ \NS(\eps,d) \leq \frac{1}{m} \AS(m,d) \mper$$
\end{theorem}

\theoremref{thm:boolns} follows immediately from this reduction along
with our bounds on Boolean average sensitivity (Theorems
\ref{thm:boolas} and \ref{thm:boolas2}), so it remains for us to prove
\theoremref{thm:reduction}.


\subsection{Proof of \theoremref{thm:reduction}}
\label{sec:prasadreduction}
\ignore{
Let $\NS(\eps,d)$ denote the maximum noise sensitivity of a degree
$d$-PTF at a noise rate of $\eps$.

The following theorem inspired by the proof of noise sensitivity of halfspaces by Peres \cite{Peres:04}, establishes a
reduction from upper bounds on average sensitivity of PTFs to corresponding bounds on their noise sensitivity.
\begin{theorem} \label{thm:reduction}
For all $0 \leq \eps \leq 1$ if $m = \lfloor \frac{1}{\eps} \rfloor$
then,
$$ \NS(\eps,d) \leq \frac{1}{m} \AS(m,d) \mper$$
\end{theorem}}

Let $f(x) = \sgn(p(x))$ be a degee $d$-PTF.  Let us denote $\delta =
\frac{1}{m}$.  As $\delta \geq \eps$, by the monotonicity of noise
sensitivity we have $\NS_{\eps}(f) \leq \NS_{\delta}(f)$. In the
following, we will show that $\NS_{\delta}(f) \leq \frac{1}{m}
\AS(m,d)$ which implies the intended result.  Recall that
$\NS_{\delta}(f)$ is defined as
$$ \NS_{\delta}(f) = \Pr_{x \sim_{\delta} y} \left[ f(x) \neq f(y)
\right]\mcom$$
where $x \sim_{\delta} y$ denotes that $y$ is generated by flipping
each bit of $x$ independently with probability $\delta$.
An alternate way to generate $y$ from $x$ is as follows:
\begin{itemize} \itemsep=0ex
\item[--] Sample $r \in \{1,\ldots,m\}$ uniformly at random.
\item[--] Partition the bits of $x$ into $m = \frac{1}{\delta}$ sets
$S_1,S_2,\ldots,S_m$ by
independently assigning each bit to a uniformly random set.  Formally,
a partition $\alpha$ is specified by a function $\alpha : \{1,\ldots,n\} \to
\{1,\ldots,m\}$ mapping bit locations to their partition numbers,
i.e., $i \in S_{\alpha(i)}$.  A
uniformly random partition is picked by sampling $\alpha(i)$ for each
$i \in \{1,\ldots, n\}$ uniformly at random from $\{1,\ldots,m\}$.
\item[--] Flip the bits of $x$ contained in the set $S_r$ to obtain $y$.
\end{itemize}
Each bit of $x$ belongs to the set $S_r$ independently with
probability $\frac{1}{m} = \delta$.  Therefore, the vector $y$
generated by the above procedure can equivalently be generated by
flipping each bit of $x$ with probability $\delta$.

Inspired by the above procedure, we now define an alternate equivalent
procedure to generate the pair $x \sim_{\delta} y$.
\begin{itemize}\itemsep=0ex
\item[--] Sample $a \in \bn$ uniformly at random.
\item[--] Sample a uniformly random partition $\alpha : \{1,\ldots,n\} \to
\{1,\ldots,m\}$ of the bits of $a$.
\item[--] Sample $z \in \bits^m$ uniformly at random.
\item[--] Sample $r \in \{1,\ldots,m\}$ uniformly at random.  Let
$\tilde{z} = z^{\oplus r}$ and
\begin{align*}
x_i = a_i z_{\alpha(i)}  & & y_i = a_i \tilde{z}_{\alpha(i)}
\end{align*}
\end{itemize}
Notice that $x$ is uniformly distributed in $\bn$, since both
$a$ and $z$ are uniformly distributed in $\bn$ and $\bits^m$
respectively.  Furthermore, $\tilde{z}_i = z_i$ for all $i \neq r$ and
$\tilde{z}_r = -z_r$.  Therefore, $y$ is obtained by flipping the bits
of $x$ in the coordinates belonging to the $r$\th partition.  As the
partition $\alpha$ is generated uniformly at random, this amounts to
flipping each bit of $x$ with probability exactly $\frac{1}{m} =
\delta$.

The noise sensitivity of $f$ can be rewritten as,
\begin{align*}
\NS_{\delta}(f) &= \Pr_{a,\alpha,z,r}\left[f(x) \neq f(y)\right]
\end{align*}
For a fixed choice of $a$ and $\alpha$, $f(x)$ is a function of $z$. In this light, let us define the function $f_{a,\alpha}
:\bits^m \to \bits$ for each $a,\alpha$ as $f_{a,\alpha}(z) = f(x)$. Returning to the expression for noise sensitivity we get:
\begin{align*}
  \NS_{\delta}(f) &= \Pr_{a,\alpha,z,r}\left[f_{a,\alpha}(z) \neq
    f_{a,\alpha}(\tilde{z}) \right] \\
  &= \E_{a,\alpha,z,r}\left[\one[f_{a,\alpha}(z) \neq
    f_{a,\alpha}(z^{\oplus r})] \right] \\
  &= \E_{a,\alpha,z}\left[ \frac{1}{m} \sum_{r=1}^m
    \one\left[f_{a,\alpha}(z) \neq
      f_{a,\alpha}(z^{\oplus r}) \right]\right] \\
  &= \E_{a,\alpha}\left[ \frac{1}{m} \sum_{r=1}^m \E_{z}
    \left[\one\left[f_{a,\alpha}(z) \neq f_{a,\alpha}(z^{\oplus r})
      \right]\right]\right] \mper
\end{align*}
In the above calculation, the notation $\one[E]$ refers to the
indicator function of the event $E$.  Recall that, by definition of
influences,
$$\Inf_r(f_{a,\alpha}) = \E_{z}\left[\one\left[f_{a,\alpha}(z) \neq
f_{a,\alpha}(z^{\oplus r})\right]\right] \mcom$$ for all $r$.  Thus, we can rewrite the noise sensitivity of $f$ as
\begin{equation}\label{eqn:noisestabilityinfl}
 \NS_{\delta}(f) = \E_{a,\alpha}\left[\frac{1}{m} \sum_{r=1}^m
\Inf_r(f_{a,\alpha})\right] =
\frac{1}{m}\E_{a,\alpha}\left[\Inf(f_{a,\alpha})\right]
\mper
\end{equation}
We claim that $f_{a,\alpha}$ is a degree $d$-PTF in $m$ variables.  To
see this observe that
$$ f_{a,\alpha}(z) = \sgn(p(x_1,\ldots,x_n)) = \sgn\left(p(a_1
z_{\alpha(1)},\ldots,a_{n}z_{\alpha(n)})\right), $$
which for  a fixed choice of $a,\alpha$ is a degree $d$-PTF in $z$.
Consequently, by definition of $\AS(m,d)$ we have $\Inf(f_{a,\alpha})
\leq \AS(m,d)$ for all $a$ and $\alpha$.  Using this in
\eqref{eqn:noisestabilityinfl}, the result follows.

\section{Application to Agnostic Learning}
 \label{sec:learn}

In this section, we outline the applications of the noise sensitivity bounds presented in this work to agnostic learning of
PTFs. Specifically, we will present the proofs of \theoremref{thm:learnuniform} and \theoremref{thm:learngaussian}.  To begin
with, we recall the main theorem of \cite{KKMS:08} about the $L_1$ polynomial regression algorithm:

\begin{theorem} \label{thm:l1}
Let $\D$ be a distribution over $X \times \{-1,1\}$ (where $X \subseteq \R^n$) which has marginal $\D_X$ over $X.$  Let $\calC$
be a class of Boolean-valued functions over $X$ such that for every $f \in \calC,$ there is a degree-$d$ polynomial
$p(x_1,\dots,x_n)$ such that $\E_{x \sim \D_X}[(p(x)-f(x))^2] \leq \eps^2.$  Then given independent draws from $\D$, the $L_1$
polynomial regression algorithm runs in time poly$(n^d,1/\eps,\log(1/\delta))$ and with probability $1 - \delta$ outputs a
hypothesis $h: X \times \{-1,1\}$ such that $ \Pr_{(x,y) \sim \D}[h(x) \neq y] \leq \opt + \eps, $ where $\opt = \min_{f \in
\calC} \Pr_{(x,y) \sim \D}[f(x) \neq y].$
\end{theorem}

We first consider the case where $\D_X$ is the uniform distribution
over the $n$-dimensional Boolean hypercube $\{-1,1\}^n$.  Klivans et
al. \cite{KOS:04} observed that Boolean noise sensitivity bounds are
easily shown to imply the existence of low-degree polynomial
approximators in the $L_2$ norm under the uniform distribution on
$\{-1,1\}^n$:

\begin{fact} \label{fact:KOS04}
For any Boolean function $f: \{-1,1\}^n \to \{-1,1\}$ and any value
$0 \leq \gamma < 1/2$, there is a polynomial $p(x)$ of degree at
most $d = 1/\gamma$ such that $\E[(p(x)-f(x))^2] \leq {\frac 2 {1 -
e^{-2}}} \NS_\gamma(f).$
\end{fact}

\theoremref{thm:learnuniform} follows directly from
\theoremref{thm:l1}, Fact~\ref{fact:KOS04} and
\theoremref{thm:boolns}.

Next we turn to the case where $\D_X$ is the ${\calN}(0,I_n)$ distribution over $\R^n$.  In \cite{KOS:08} observed that using
entirely similar arguments to the Boolean case, Gaussian noise sensitivity bounds imply the existence of low-degree polynomial
approximators in the $L_2$ norm:

\begin{fact} \label{fact:KOS08}
For any Boolean function $f: \{-1,1\}^n \to \{-1,1\}$ and any value
$0 \leq \gamma < 1/2$, there is a polynomial $p(x)$ of degree at
most $d = 1/\gamma$ such that $\E_{\calG \sim
\calN(0,I_n)}[(p(\calG)-f(\calG))^2] \leq {\frac 2 {1 - e^{-1}}}
\GNS_\gamma(f).$
\end{fact}

For the special case of learning under the standard multivariate Gaussian $\N^n$, \theoremref{thm:learngaussian} follows
directly from \theoremref{thm:l1}, Fact~\ref{fact:KOS08} and \theoremref{thm:gaussns}. Since our results hold for \emph{all}
degree-$d$ PTFs, the extension to arbitrary Gaussian distributions follows exactly as described in Appendix~C of \cite{KOS:08}.

\section{Discussion} \label{sec:discussion}

An obvious question left open by this work is to actually resolve
the Gotsman-Linial conjecture and show that every degree-$d$ PTF
over $\{-1,1\}^n$ has average sensitivity at most $O(d\sqrt{n}).$
\cite{GS:09} show that this would have interesting implications in
computational learning theory beyond the obvious strengthenings of
the agnostic learning results presented in this paper.

In this section we observe (Proposition~\ref{obs:equiv}) that this conjecture is in fact equivalent to a strong upper bound on
the Boolean noise sensitivity of degree-$d$ PTFs.  We further point out (Proposition~\ref{obs:gausslessthanbool}) that Gaussian
noise sensitivity of degree-$d$ PTFs is upper bounded by Boolean noise sensitivity.  Thus, we propose working on improved upper
bounds for the Gaussian noise sensitivity of degree-$d$ PTFs as a  preliminary -- in fact, necessary -- step to settling the
Gotsman-Linial conjecture. \ignore{(In fact, the observations below imply that this ``preliminary'' step is \emph{necessary} to
obtain the Gotsman-Linial conjecture: if every degree-$d$ PTF has average sensitivity at most $O(d \sqrt{n})$, then the
Gaussian noise sensitivity of every degree-$d$ PTF is at most $O(d \sqrt{\eps}).$)}

\begin{prop} \label{obs:equiv}
The following two statements are equivalent:
\begin{enumerate}
\item  Every degree-$d$ PTF over $\{-1,1\}^n$ has $\AS(f) \leq O(d
\sqrt{n}).$

\item Every degree-$d$ PTF over $\{-1,1\}^n$ has $\NS_\eps(f) \leq O(d
\sqrt{\eps})$ for all $\eps.$
\end{enumerate}
\end{prop}

\begin{proof}

1) $\Rightarrow$ 2):  This follows immediately from \tref{reduction}

\medskip

2) $\Rightarrow$ 1):  Let $f=\sign(p)$ be a degree-$d$ PTF.  We have
\begin{eqnarray*}
\NS_{1/n}(f) &=& \Pr_{x,y}[f(x) \neq f(y)]\\
&=&\sum_{k=0}^n \Pr_{x,y}[f(x) \neq f(y) \ | \ y\text{ flips $k$ of
$x$'s bits}] \cdot \Pr_{x,y}[y\text{ flips $k$ of $x$'s bits}]\\
&\geq& \Pr_{x,y}[f(x) \neq f(y) \ | \ y\text{ flips 1 of $x$'s
bits}] \cdot \Pr_{x,y}[y\text{ flips 1 of $x$'s bits}]\\
&\geq& (1/n)\AS(f) \cdot \Theta(1),
\end{eqnarray*}
where the last inequality holds because at noise rate $1/n$, there
is constant probability that $y$ flips exactly 1 of $x$'s bits, and
conditioned on this taking place, the probability that $f(x) \neq
f(y)$ is exactly $\AS(f)/n.$  Taking $\eps = 1/n$ in 2) and
rearranging, we get 1).
\end{proof}

\begin{prop} \label{obs:gausslessthanbool}
Let $\NS(\eps,d)$ and $\GNS_{\eps,d}$ denote the maximum noise
sensitivity of a degree $d$ PTF in the Boolean and Gaussian domains
respectively.  For all $\eps$ and $d$, we have
\[
\NS(\eps,d) \geq \GNS(\eps,d) .
\]
\end{prop}

\begin{proof}
Consider a degree-$d$ PTF $f = \sgn(p(x))$ in the Gaussian setting.
\ignore{Here the polynomial $p(x_1,\ldots,x_n)$ is in general not
multilinear.} We will define a sequence of degree-$d$ PTFs
$\{h_{k}\}_{k=1}^{\infty}$ over the Boolean domain.  The function
$h_k : \{-1,1\}^{nk} \to \{-1,1\}$ is on $nk$ input bits
$\{y_{i}^{(j)} | i \in [n], j \in [k]\}$ and is given by,
$$ h_k(y_{1}^{(1)},y_{1}^{(2)},\ldots,y_{n}^{(k)}) \eqdef
\sgn\left(p\left(\frac{\sum_{j \in [k]}
y_{1}^{(j)}}{\sqrt{k}},\frac{\sum_{j \in [k]}
y_{2}^{(j)}}{\sqrt{k}},\dots, \frac{\sum_{j \in [k]}
y_{n}^{(j)}}{\sqrt{k}} \right)\right) .$$

By the Central Limit Theorem, the normalized sum $\frac{\sum_{j \in
[k]} y_{i}^{(j)}}{\sqrt{k}}$ of $k$ independent random values from
$\{-1,1\}$,  tends to in distribution to the normal distribution
$\N(0,1)$ as $k \rightarrow \infty$.  Intuitively, this implies that
as $k \rightarrow \infty$, among other things the Boolean noise
sensitivity of $h_k$ approaches the noise sensitivity of $f$.
However, since $h_k$ is a Boolean PTF its noise sensitivity is
bounded by $\NS(\eps,d)$.

We now present the details of the above argument. Consider the
random variables $y = (y_{1},\dots,y_{n}), \tilde{y} =
(\tilde{y}_{1},\dots,\tilde{y}_n) \in \{-1,1\}^{n}$ generated by
setting each $y_i$ to an uniform random value in $\{-1,1\}$ and
$\tilde{y}_i$ as
$$\tilde{y}_i = \begin{cases} y_i & \text{ with probability } 1-\eps\\
                \text{uniform value in } \{-1,1\} & \text{ with
probability~} \eps.
                \end{cases}$$
It is clear that $\E[y_{i}\tilde{y}_i] = 1-\eps$ for all $i \in [n]$
and all other pairwise correlations are $0$.  Let $\{
(y^{(1)},\tilde{y}^{(1)}),\dots,(y^{(k)},\tilde{y}^{(k)})\}$ be $k$
independent samples of $(y,\tilde{y})$.  By definition of Boolean
noise sensitivity,
\begin{align*}
 \NS_{\eps}(h_k) & = \Pr[h_{k}(y) \neq h_{k}(\tilde{y})] \\
                & = \Pr\left[p\left(\frac{\sum_{j \in [k]}
y^{(j)}}{\sqrt{k}} \right) \cdot p\left(\frac{\sum_{j \in [k]}
\tilde{y}^{(j)}}{\sqrt{k}} \right) \leq 0 \right].
\end{align*}

Let $x \sim \N^n, z \sim \N^n$ be independent and let $\tilde{x} = \alpha x + \beta z$, with $\alpha = 1-\eps$ and $\beta =
\sqrt{2\eps-\eps^2}$.  By the Multidimensional Central Limit Theorem~\cite{Feller}, as $k \rightarrow \infty$  we have the
following convergence in distribution,
$$\left(\frac{\sum_{j \in [k]}
y^{(j)}}{\sqrt{k}},\frac{\sum_{j \in [k]}
\tilde{y}^{(j)}}{\sqrt{k}}\right) \xrightarrow{\mathcal{D}}
(x,\tilde{x}).$$ Since the function $a(x,\tilde{x}) = p(x) \cdot
p(\tilde{x})$ is a continous function we get
\begin{eqnarray*}
\lim_{k \rightarrow \infty} \NS_{\eps}(h_k) &=& \lim_{k \rightarrow
\infty} \Pr\left[p\left(\frac{\sum_{j \in [k]} y^{(j)}}{\sqrt{k}}
\right) \cdot p\left(\frac{\sum_{j \in [k]}
\tilde{y}^{(j)}}{\sqrt{k}} \right) \leq 0 \right]\\
&=& \Pr_{x,\tilde{x}}[p(x)p(\tilde{x}) \leq 0]\\
&=& \GNS_{\eps}(f)
\end{eqnarray*}
and the result is proved.
\end{proof}


\bibliographystyle{alpha} \bibliography{allrefs}

\appendix

\section{Basics of Hermite Analysis} \label{ap:hermite}

Here we briefly review the basics of Hermite analysis over $\R^n$
under the distribution $\N^n.$  The reader who is unfamiliar with
Hermite analysis should note the many similarities to Fourier
analysis over $\{-1,1\}^n.$

We work within $L^2(\R^n, \calN^n)$, the vector space of all
functions $f : \R^n \to \R$ such that $\E_{x \sim \calN^n}[f(x)^2] <
\infty$. This is an inner product space under the inner product
\[
\la f, g \ra = \Ex_{x \sim \calN^n} [f(x)g(x)].
\]
This inner product space has a complete orthonormal basis given by
the \emph{Hermite polynomials}. In the case $n = 1$, this basis is
the sequence of polynomials
\[
h_0(x) = 1, \quad h_1(x) = x, \quad h_2(x) = \frac{x^2-1}{\sqrt{2}},
\quad h_3(x) = \frac{x^3 - 3x}{\sqrt{6}}, \quad \dots,
\]
\[
h_j(x) = \frac{\sqrt{j!}}{(j-0)!0!2^0} x^j -
\frac{\sqrt{j!}}{(j-2)!1!2^1}x^{j-2} +
\frac{\sqrt{j!}}{(j-4)!2!2^2}x^{j-4} - \frac{\sqrt{j!}}{(j-6)!3!2^3}
x^{j-6} + \cdots
 \] which may equivalently be defined by
\[
h_j(x) = \frac{(-1)^d}{\sqrt{d!} \exp(-x^2/2)} \cdot
\frac{d^j}{dx^j} \exp(-x^2/2).
\]
We note that $h_d(x)$ is a polynomial of degree $d.$ For general
$n$, the basis for $L^2(\R^n, \calN^n)$ is formed by all products of
these polynomials, one for each coordinate.  In other words, for
each $n$-tuple $S \in \mathbb{N}^n$ we define the $n$-variate
Hermite polynomial $H_S : \R^n \to \R$ by
\[
H_S(x) = \prod_{i=1}^n h_{S_i}(x_i);
\]
then the collection $(H_S)_{S \in \N^n}$ is a complete orthonormal
basis for the inner product space.  By orthonormal we mean that
\[
\la H_S, H_T \ra = \begin{cases} 1 & \text{if $S = T$,} \\ 0 &
\text{if $S \neq T$.} \end{cases}
\]
By complete, we mean that every function $f \in L^2$ can be uniquely
expressed as
\[
f(x) = \sum_{S \in \mathbb{N}^n} \widehat{f}(S) H_S(x),
\]
where the coefficients $\widehat{f}(S)$ are real numbers and the
infinite sum converges in the sense that
\[
\lim_{d \to \infty} \E\left[\left(f(x) - \sum_{|S| \leq d} c_S
H_S(x)\right)^2\right] = 0;
\]
here we have used the notation
\[
|S| = \sum_{i=1}^n S_i,
\]
which is also the total degree of $H_S(x)$ as a polynomial.\\

We call $\widehat{f}(S)$ the \emph{$S$ Hermite coefficient of $f$}.
By orthonormality of the basis $(H_S)_{S \in \N^n}$, we have the
following:
\[
\widehat{f}(S) = \la f, H_S \ra = \E[f(x) H_S(x)];
\]
\[
\|f\|_2^2 \eqdef \la f, f \ra = \sum_{S \in \mathbb{N}^n}
\widehat{f}(S)^2 \qquad \text{(``Parseval's identity'')};
\]
\[
\la f, g \ra = \sum_{S \in \mathbb{N}^n}
\widehat{f}(S)\widehat{g}(S) \qquad \text{(``Plancherel's
identity'')}.
\]
In particular, if $f : \R^n \to \{-1,1\}$, then $\sum_S
\widehat{f}(S)^2 = 1$.

Using the definition of influence from Section~\ref{sec:basicdef}, it is not difficult to show that for any $f: \R^n \to \R$
and any $i \in [n]$, we have $\GI_i(f) = \sum_{S: S_i > 0} \widehat{f}(S)^2$ (see e.g. Lecture~4 of \cite{Mossel:05}).

\end{document}